\documentclass[a4paper,11pt]{amsart}

\usepackage[hidelinks]{hyperref}

\usepackage{url}
\usepackage{cite}

\usepackage{amssymb}
\usepackage{color}

\usepackage{bbm}

\usepackage[normalem]{ulem}
\usepackage[utf8]{inputenc}
\usepackage{graphicx}

\usepackage{soul}

\usepackage{todonotes}

\presetkeys{todonotes}{fancyline, color=green!40, size=\tiny}{}

\newcommand{\ptclean}[1]{}

\newcommand{\fpsi}{\psi}

\newcommand{\deltazgriem}{h}
\newcommand{\deltagriem}{\deltazgriem}

\newcommand{\ISstar}{{\mathbb{\efindex}^\star_S}}
\newcommand{\IVstar}{{\mathbb{\efindex}^\star_V}}
\newcommand{\ISdiamond}{{\mathbb{\efindex}^\diamond_S}}
\newcommand{\IVdiamond}{{\mathbb{\efindex}^\diamond_V}}

\newcommand{\hTstar}[1]{{\deltazgriem^{T\star}_{#1}}}

\newcommand{\hSstar}[1]{{\deltazgriem^{S\star}_{#1}}}
\newcommand{\YSstar}[1]{{Y^{S\star}_{#1}}}
\newcommand{\hVstar}[1]{{\deltazgriem^{V\star}_{#1}}}
\newcommand{\YVstar}[1]{{Y^{V\star}_{#1}}}
\newcommand{\hSdiamond}[1]{{\deltazgriem^{S\diamond}_{#1}}}
\newcommand{\YSdiamond}[1]{{Y^{S\diamond}_{#1}}}
\newcommand{\Ydiamondm}[1]{{Y^{\diamond-}_{#1}}}
\newcommand{\hSs}[1]{\hSdiamond{#1}}

\newcommand{\hstar}[1]{{\deltazgriem^{\star}_{#1}}}
\newcommand{\Ystar}[1]{{Y^{\star}_{#1}}}

\newcommand{\hdiamond}[1]{{\deltazgriem^{\diamond}_{#1}}}
\newcommand{\Ydiamond}[1]{{Y^{\diamond}_{#1}}}

\newcommand{\hVdiamond}[1]{{\deltazgriem^{V\diamond}_{#1}}}

\newcommand{\hVs}[1]{\hVdiamond{#1}}

\newcommand{\VtwoOrF}{{f}}

 \newcommand{\mvI}{\mathbb{V}^I_i}
 \newcommand{\mvIj}{\mathbb{V}^I_j}
 \newcommand{\mVI}{\mvI}
 \newcommand{\mVIj}{\mvIj}
 \newcommand{\HTVI}{{H^V_{T,I}}}
 \newcommand{\XaI}{{X_{a,I}}}

\newcommand{\faVI}{{f_{a,I}^V}}
\newcommand{\faSI}{{f_{a,I}^S}}
\newcommand{\fabSI}{{f_{ab,I}^S}}
\newcommand{\HTTI}{{H_{T,I}^T}}
\newcommand{\HTSI}{{H_{T,I}^S}}
\newcommand{\HLSI}{{H_{L,I}^S}}

\newcommand{\lambdassquare}{\frac{s^2}{\ell^2}}
\newcommand{\lambdaonesquare}{\ell^{-2}}
\newcommand{\lambdadtsquare}{\ell^{-2} dt^2}
\newcommand{\lambdarsquare}{\frac{ r^2}{\ell^2}}
\newcommand{\lambdarzerosquare}{\frac{r^2_0}{\ell^2}}

\newcommand{\kappahere}{{{}K}}
\newcommand{\ourell}{{}\ell}
\newcommand{\mhere}{{}\mu}
\newcommand{\ourellnormal}{{}l}
\newcommand{\efindex}{{}I}
\newcommand{\modeindex}{\efindex}
\newcommand{\Nman}{{}^{n}N}
\newcommand{\theirn}{{n}}
\newcommand{\ourn}{{n+1}}
\newcommand{\ourntwo}{{n-1}}
\newcommand{\rreg}{\check  r}
\newcommand{\rzero}{r_0}

\newcommand{\rzeropower}[1]{r_0^{#1}}

\newcommand{\rhoreg}{ \rho}

\newcommand{\qedskip}{\hfill $\Box$ \medskip}

\newcommand{\bean}{\begin{eqnarray}\nn}

\newcommand{\nn}{\nonumber}

\newcommand{\Eq}[1]{Equation~\eqref{#1}}

\newcommand{\zglorentz}{{ \mathring{\glorentz}}}
\newcommand{\glorentz}{{ {\mathbf g}}}

\newcommand{\griem}{{ {\mathfrak g}}}

\newcommand{\zgriem}{{ \mathring{\mathfrak g}}}
\newcommand{\ztgriem}{{ {}^2\mathring{\mathfrak g}}}

\DeclareFontFamily{OT1}{rsfs}{}
\DeclareFontShape{OT1}{rsfs}{m}{n}{ <-7> rsfs5 <7-10> rsfs7 <10->
rsfs10}{} \DeclareMathAlphabet{\mathscr}{OT1}{rsfs}{m}{n}

\newcommand{\eq}[1]{\eqref{#1}}

\newcommand{\bel}[1]{\begin{equation}\label{#1}}
\newcommand{\beal}[1]{\begin{eqnarray}\label{#1}}
\newcommand{\beadl}[1]{\begin{deqarr}\label{#1}}
\newcommand{\eeadl}[1]{\arrlabel{#1}\end{deqarr}}
\newcommand{\eeal}[1]{\label{#1}\end{eqnarray}}
\newcommand{\eead}[1]{\end{deqarr}}
\newcommand{\eea}{\end{eqnarray}}
\newcommand{\eeaa}{\end{eqnarray*}}

\newcommand{\be}{\begin{equation}}
\newcommand{\ee}{\end{equation}}

\DeclareFontFamily{OT1}{rsfs}{}
\DeclareFontShape{OT1}{rsfs}{m}{n}{ <-7> rsfs5 <7-10> rsfs7 <10->
rsfs10}{} \DeclareMathAlphabet{\mycal}{OT1}{rsfs}{m}{n}

\newcommand{\mcL}{{\mycal L}}
\newcounter{mnotecount}[section]

\newcommand{\N}{{\mathbb{N}}}
\newcommand{\Z}{{\mathbb{Z}}}

\newcommand{\rmnote}[1]{}

\newcommand{\Ric}{\operatorname{Ric}}

\def\mysavedown#1{\edef\mysubs{\mysubs#1}}
\def\mysaveup#1{\edef\mysups{\mysups#1}}
\def\mydown#1{{\mytensor}_{\vphantom{\mysubs}#1}}
\def\myup#1{{\mytensor}^{\vphantom{\mysups}#1}}
\def\tensor#1#2{
  #1
  \def\mytensor{\vphantom{#1}}
  \def\mysubs{\relax}
  \def\mysups{\relax}
  \let\down=\mysavedown
  \let\up=\mysaveup
  #2
  \let\down=\mydown
  \let\up=\myup
  #2
  }

\newcommand{\Tr}{\operatorname{Tr}}
\newcommand{\grav}{\operatorname{grav}}

\newcommand{\R}{\mathbb R}

\renewcommand{\to}{\rightarrow}
\renewcommand{\div}{\operatorname{div}}

\newcommand{\half}{{\tfrac12}}

\renewcommand{\epsilon}{\varepsilon}
\renewcommand{\hat}{\widehat}

\def\crn#1#2{{\vcenter{\vbox{
        \hbox{\kern#2pt \vrule width.#2pt height#1pt
           }
          \hrule height.#2pt}}}}

\renewcommand{\H}{\mathbb H}

\renewcommand{\hbar}{{\overline h}}

\newcommand{\pre}[2]{{{\vphantom{#2}}^{#1}}\kern-.2ex{#2}}

\theoremstyle{plain}
\newtheorem{theorem}{\sc Theorem}[section]
\newtheorem{lemma}[theorem] {\sc Lemma}
\newtheorem{proposition}[theorem]{\sc Proposition}
\newtheorem{corollary}[theorem] {\sc Corollary}
\newtheorem{bigtheorem} {\sc Theorem}[section]
\newtheorem{conjecture}[theorem] {\sc Conjecture}
\newtheorem{Conjecture}[bigtheorem] {\sc Conjecture}

\theoremstyle{definition}

\newtheorem{Remark}[theorem]{\sc  Remark\rm}
\newtheorem{remark}[theorem]{\sc  Remark\rm}

\numberwithin{equation}{section}

\date{\today}

\begin{document}

\title[On non-degenerate Riemannian Kottler metrics] {On non-degeneracy of Riemannian  Schwarzschild-anti de Sitter metrics}
\thanks{Preprint UWThPh-2017-30}

\author[P.T. Chru\'sciel]{Piotr T.~Chru\'sciel}

\address{Piotr
T.~Chru\'sciel, Faculty of Physics and Erwin Schr\"odinger Institute, University of Vienna, Boltzmanngasse 5, A1090 Wien, Austria}
\email{piotr.chrusciel@univie.ac.at} \urladdr{http://homepage.univie.ac.at/piotr.chrusciel/}

\author[E. Delay]{Erwann Delay}
 \address{Erwann Delay, Laboratoire de math\'ematiques d'Avignon, F-84 916 AVIGNON, et FRUMAM - CNRS
 Marseille, France, and Erwin Schr\"odinger Institute, University of Vienna} \email{Erwann.Delay@univ-avignon.fr}
\urladdr{http://www.math.univ-avignon.fr/}

\author[P. Klinger]{Paul Klinger}

\address{Paul Klinger,  Faculty of Physics and Erwin Schr\"odinger Institute, University of Vienna, Boltzmanngasse 5, A1090 Wien, Austria}
\email{paul.klinger@univie.ac.at}

\begin{abstract}
We prove that the $TT$-gauge-fixed linearised Einstein operator is non-degenerate for Riemannian Kottler (``Schwarzschild-anti de Sitter'') metrics with dimension- and topology-dependent ranges of mass parameter. We provide evidence that this remains true for all such metrics except the spherical ones with a critical mass.
\end{abstract}

\maketitle

\tableofcontents
\section{Introduction}
 \label{section:intro}

There is currently considerable interest in the literature in spacetimes with a negative cosmological constant. In particular one is interested in existence of stationary black hole solutions of the Einstein equations with $
\Lambda<0$, with or without sources, and in properties thereof. Many such solutions have been constructed numerically, e.g.~\cite{WinstanleyEYM,BjorakerHosotani,RaduTchrakian,OkuyamaMaeda,BBSLMR}.

In~\cite{ACD2,ChDelayKlingerBH} we showed how to construct infinite dimensional families of  non-singular stationary
black hole spacetimes, solutions of the Einstein  equations with a negative cosmological constant in vacuum or with various matter sources, assuming that a suitable linearised operator was an isomorphism.  In~\cite[Proposition~D.2]{ACD2} we proved the isomorphism property at Kottler solutions with negatively curved sections of conformal infinity for some
dimension-dependent explicit ranges of the mass parameter (the whole range of masses in space-time dimension four,  and the whole range of negative masses in all dimensions).
 However, the case of flat or positively curved conformal infinity has  been open so far.
The aim of this work is to prove the optimal result for all topologies in spacetime dimension four, and to provide partial answers to this problem in higher dimensions.
This extends immediately the applicability of  the existence theorems of~\cite{ACD2,ChDelayKlingerBH} to the topologies and dimensions covered here.

Indeed, we prove:

\begin{bigtheorem}
 \label{T15X17.1}
Let us denote by $P_L$ the linearisation,  at Riemannian  Kottler metrics \eq{21III17.1} with negative cosmological constant, of the $TT$-gauge-fixed Einstein operator. Then:

\begin{enumerate}
  \item $P_L$
 has no $L^2$-kernel in spacetime dimension $n+2=4$  except for spherical black holes with mass parameter
\bel{9VII17.1+5}
 \mhere= \frac{n}{n+1}
 \left(\ell
   \sqrt{\frac{n-1}{n+1}}\right)^{n-1}
 \,.
\ee
  \item  $P_L$
  has no $L^2$-kernel for  toroidal and higher genus black holes for all mass parameters $\mhere$ in dimension $n=2$, as well as for open ranges of parameters $\mu\in (0,\mu(n))$ for $n>2$ and $K=0$, where $\mu(n)>0$ solves a polynomial equation; cf.\ Table~\ref{t:mulimits} in low dimensions.
\end{enumerate}
\end{bigtheorem}

Theorem~\ref{T15X17.1} is a rewording of Theorem~\ref{T8IX17.1} below.

In order to prove an equivalent of part (2) of Theorem~\ref{T15X17.1} for black holes in higher-dimensions with $K\in\{1,-1\}$ it remains to establish a higher-dimensional topology-independent linearised-Birkhoff-type theorem and to treat the $l=1$ modes for $K=1$ in dimension $n>2$. This is done in the accompanying paper \cite{Klinger2018}.

We also propose a scheme, supported by numerical evidence, that could establish an equivalent of part (1) of Theorem~\ref{T15X17.1} for all topologies and spacetime dimensions, except for the exceptional value \eqref{9VII17.1+5} of the mass in the spherical case. We thus conjecture:

\begin{Conjecture}
  \label{C15X17.21+}
 $P_L$ has no $L^2$-kernel except if $K=1$ and  $\mhere $ is given by \eq{9VII17.1+5}.
\end{Conjecture}

In view of the results in~\cite{Klinger2018}. in order to establish the conjecture one would need to prove global existence for the ODE we integrate numerically.

We emphasise that our four-dimensional result, namely point (1) of Theorem~\ref{T15X17.1} holds regardless of the genus of the black hole.
The proof in~\cite[Proposition~D.2]{ACD2} of the higher-genus case $K=-1$  of (1) uses a completely different method, and leads to restricted ranges of masses in higher dimension. Our method here provides an alternative argument
which, at the level of numerical tests, covers all masses.

Now, the idea in~\cite{ACD2,ChDelayKlingerBH} is to show that the construction of stationary Lorentzian solutions near a known static $d$-dimensional metric $\zglorentz$ can be carried-out using the implicit function theorem near a Riemannian
$d$ dimensional partner metric $\zgriem$. Consider then the operator $P_L$  obtained by linearising the Einstein equations at the ``Wick rotated'' metric $\zgriem$ associated to $\zglorentz$ and imposing the $TT$ gauge.
 The construction of~\cite{ACD2,ChDelayKlingerBH} applies if $P_L$ has no $L^2$-kernel; we then say that the Riemannian metric $\zgriem$ is non-degenerate.
We show here  that  non-degeneracy holds for Kottler
metrics with spherical black hole horizons in four spacetime dimension, and for toroidal black hole horizons in all dimensions, for wide ranges of masses.

Some comments on the proof are in order. For this, let $\deltagriem$ be an element of the $L^2$-kernel of the operator $P_L$ defined above. Our proof relies heavily on the remarkable construction of master functions of Ishibashi and Kodama~\cite{KodamaIshibashiMaster}. Indeed, $\deltagriem$ can be decomposed into a sum of eigentensors of relevant operators, which we will refer to as modes. The modes split into two families, which we will refer to as exceptional modes and master modes. The master modes of $\deltagriem$ are controlled by the Ishibashi-Kodama master functions, which we show to be zero under the conditions above. We use the vanishing of the master modes to show existence of a vector field
{\em with controlled asymptotic behaviour}
which can be used to gauge-away the relevant part of $\deltagriem$. Likewise we show that the special modes are pure gauge in the last,
controlled, sense. Adding the associated vector fields  and using the TT-gauge condition, together with suitable Birkhoff-type linearised theorems, we establish that the kernel is trivial in the cases listed in Theorem~\ref{T15X17.1}.


We note that a non-trivial kernel of $P_L$ implies existence of growing linear modes  in the corresponding spacetime,  except perhaps for variations in the direction of nearby stationary metrics. So proving that $P_L$ has no $L^2$-kernel has implications for the dynamics of the associated Lorentzian solutions.
However, non-degeneracy does not imply linear dynamical stability on the Lorentzian side, since non-degeneracy only excludes modes with a specific spectrum of frequencies.

\section{Riemannian Kottler metrics}
 \label{sA7IV17.1}

We start with a review of the Riemannian counterparts of the
``generalised Kottler~\cite{Kottler} metrics'', also known as ``Schwarzschild-anti de Sitter metrics'', or ``Birmingham metrics''~\cite{Birmingham}.

In what follows we will be making extensive use of~\cite{KodamaIshibashiMaster}, in order to be consistent as much as possible with their notations we set
$$
 n= d-2
 \,,
$$
where $d$ is the \emph{dimension of spacetime}.
(The reader is warned that this \emph{does not} coincide with our notations in~\cite{ACD2,ChDelayKlingerBH}, where $n$ denotes space dimension $d-1$.)

The metrics of interest read
\bel{21III17.1}
 \zgriem = \bigg(\underbrace{\lambdarsquare  +\kappahere  - \frac{2\mhere }{r^{\ourntwo}}}_{=:\VtwoOrF(r)  }\bigg) dt^2
   + \frac{dr^2} {\lambdarsquare  +\kappahere  - \frac{2\mhere }{r^{\ourntwo}}} + r^2 h_\kappahere
 \,,
\ee
where  $\ell $ is a constant related to the cosmological constant
$\Lambda<0$ by the formula
$$
 \ell  =    \sqrt{-\frac{n(n+1)}{2\Lambda}}  > 0
 \,,
$$
$\mhere $ is a real constant, and $\kappahere \in\{0, \pm 1\}$.
Furthermore, $h_\kappahere$ is a metric on
an $\theirn$-dimensional Einstein manifold  $\Nman $.
The metric $ \zgriem$ is Einstein if the Ricci tensor of $h_\kappahere$ equals $(n-1) \kappahere  h_\kappahere $.
  This will be assumed in what follows.

We require that $\VtwoOrF$ has  positive zeros, we denote by  $\rzero >0$
   the largest such zero, which we assume to be of first order (as will necessarily be the case if $\kappahere  \ge 0$ and $\mhere>0$):
\bel{11IV17.1}
 \lambdarzerosquare  +\kappahere  - \frac{2\mhere }{\rzeropower{\ourntwo}}= 0
 \,.
\ee
After introducing a new coordinate $\rreg $ by the formula
\bel{21III17.2}
 \rreg  (r) = \int_{\rzero }^r \frac 1 {\sqrt{\lambdassquare  +\kappahere  - \frac{2\mhere }{s^{\ourntwo}}}} ds
 \,,
\ee
one can rewrite  the metric \eq{21III17.1} as
\bel{21III17.3}
  \zgriem= d \rreg ^2 +  \rreg ^2 H(\rreg ) dt^2  + r^2 h_\kappahere
 \,,
\ee
where $H$ is obtained by dividing $g_{tt}$ by $\rreg ^2$. Elementary analysis, using the fact that $\rzero $ is a simple zero of $F$, shows that
$$
 H(0) = \frac{\VtwoOrF'(\rzero )^2} 4
  =
   \left(\frac{  (n+1)  \ourell^{-2}  \rzero^2 + K (n-1)}{2 \rzero}\right)^2
 \,.
$$
This implies that a periodic identification of $t$ with period
\bel{10IV17.2}
 T:= \frac {4 \pi} {\VtwoOrF'(\rzero )}
\ee
guarantees that $   d \rreg ^2 +  \rreg ^2 H(\rreg ) dt^2$ is a smooth metric on $\R^2$ with a rotation axis at $\rreg =0$. As a result, \eq{21III17.3} defines a smooth Riemannian metric on
\bel{7IV17.5}
 M:= \R^2 \times {}  \Nman
 \,.
\ee
%
%

The metric  \eq{21III17.1} can be smoothly conformally compactified by   introducing, for large $r$, a coordinate $\rhoreg:=1/r$ and rescaling:
\bel{21III17.7}
 \rhoreg ^2  \zgriem = \big( \lambdaonesquare  +\kappahere \rhoreg ^2  -  {2\mhere }{\rhoreg ^{\ourn }} \big) dt^2
   + \frac{d\rhoreg ^2} {\lambdaonesquare  +\kappahere  \rhoreg ^2  -  {2\mhere }{\rhoreg ^{\ourn }}} + h_\kappahere
 \,.
\ee
%
Hence, the $(\ourn)$-dimensional
 conformal boundary $\partial M:=\{\rhoreg =0\}  $ of $M$ is diffeomorphic to  $S^1 \times {} \Nman $, with conformal metric
\bel{21III17.6}
  \lambdadtsquare    + h_\kappahere
 \,.
\ee

As already pointed out, it has been shown in~\cite[Appendix~D]{ACD2}
that all such solutions with $\ourn=3$ and $\kappahere =-1$ are non-degenerate. Furthermore, in that reference some non-degenerate families of higher-dimensional such solutions have been described explicitly.  We thus conjecture that the ranges of mass parameters there are not sharp, and that the arguments proposed here can be used to establish non-degeneracy for all values of $\mhere$ when $K\le 0$.

In the analysis of the master equations below we assume that $(\Nman,h_\kappahere )$ is a (locally) maximally symmetric compact manifold. When $K\ge 0$ the hypothesis that $\Lambda<0$ and that of existence of a black hole with $r_0>0$ implies
$$
 \mhere > 0
 \,.
$$
When $K=-1$ one needs instead
\bel{18X17.31}
 \mu > \mhere_{\min{}} :=-\frac{1
    }{n+1}\left(\frac{n+1}{\ell^2    (n-1)}\right)^{\frac{1-n}{2
   }}
    \,,
\ee
with the corresponding outermost-horizon radius $r_0={1}/{\sqrt{\frac{n+1}{\ell^2 (n-1)}}}$.

 The idea is to reduce the question of non-degeneracy to the Riemannian equivalent of the ``master equations'' of Ishibashi and Kodama~\cite{IshibashiKodamaStability,KodamaIshibashiMaster} as follows:

\begin{enumerate}
  \item rewrite the master equations in a Riemannian form;
  \item work-out the asymptotic behaviour of the master fields corresponding to  $L^2$-elements of the kernel of the shifted  Lichnerowicz operator $P_L$ defined below;
  \item prove that all associated solutions of the master equations are trivial;
  \item prove that elements of the $L^2$-kernel with trivial master fields vanish identically.
\end{enumerate}

Point (1) above is straightforward once the impressive work in~\cite{KodamaIshibashiMaster} has been carried out, but the remaining parts require some work. We note that point (4) captures the fact that the master equations contain the whole gauge-invariant information about the linearised gravitational field.

Consider, thus, an element $h=h_{\mu\nu} dx^\mu dx^\nu$ of the $L^2$-kernel of
\bel{11IV17.5}
P_L:=\Delta_L+2(n+1)
 \,,
\ee
where $\Delta_L$ is the  Lichnerowicz Laplacian,  acting on  symmetric
two-tensor fields $h_{\mu\nu}$
as~\cite[\S~1.143]{besse:einstein}
\bel{30VI17.1}
 \Delta_L
  h_{\mu \nu }=-\nabla^\alpha \nabla_\alpha h_{\mu \nu }+R_{\mu \alpha }h^\alpha {_\nu }+R_{\nu \alpha }h^\alpha {_\mu }-2R_{\mu \alpha \nu \beta}h^{\alpha \beta}\,.
\ee
It follows from~\cite{Lee:fredholm} that $|h|_{\zgriem}=O(\rho^\ourn)$, or in local coordinates near the conformal boundary,
\bel{10IV17.4}
 h_{\mu\nu} = O(\rho^{n-1})
 \,.
\ee

Let us show, first, that $h$ satisfies the linearised Einstein equations. For this,
recall that the Hodge Laplacian
acting on one-forms is defined as
\bel{DeltaH}
\Delta_H
:=
 d_g^*d+dd_g^*=\nabla^*\nabla+\Ric
\,.
\ee
If the Ricci tensor is covariantly constant,
then~\cite{Lichnerowicz61}
$$
\div\circ\Delta_L=\Delta_H\circ \div.
$$
This implies that $L^2$-elements of the kernel of $P_L$ are divergence-free:
Indeed, assume that $h$ is in the $L^2$-kernel of $P_L$ and let $u=\div h$.  We have
$$
0=\int \langle\div(\Delta_L+2(n+1))h,u\rangle
=\int \langle(\Delta_H+2(n+1))u,u\rangle,
$$
so
$$
\int |du|^2+|d^*_g u|^2+2(n+1)|u|^2=0,
$$
and $u\equiv 0$ follows. Note that there are no boundary terms in the integration-by-parts above by, e.g.,~\cite{Lee:fredholm}.

Next, recall that  we always have
$$
\Tr\circ\Delta_L=-\Delta\circ\Tr
$$
(we use the convention $\Delta= \nabla^\alpha \nabla_\alpha$).
 It follows that elements of the $L^2$-kernel of $P_L$ are trace-free.

The linearisation of the trace-shifted Ricci tensor reads
\bel{8VII17.1}
 D(\Ric+(n+1))h_{\mu\nu}=
 \frac{1}{2}\Delta_Lh_{\mu\nu}+(n+1)h_{\mu\nu}-(\div^*
 \div\grav h)_{\mu\nu},
\ee
where
\bel{8VII17.2}
 \grav h=h-\frac{1}{2}\Tr_g hg,\;\;\; (\div
 h)_\mu=-\nabla^\nu h_{\mu\nu},\;\;\;
 (\div^*w)_{\mu\nu}=\frac{1}{2}(\nabla_\mu w_\nu+\nabla_\nu w_\mu)
 \,,
\ee
and where $\Tr$ denotes the trace (note the
geometers' convention to include a negative sign in the definition of
divergence).

We have just seen that tensors in the kernel of $P_L$ are transverse and traceless. It follows from \eq{8VII17.1}-\eq{8VII17.2} that they are also in the kernel of the linearised vacuum Einstein operator.

Now, it follows immediately from the analysis in~\cite{KodamaIshibashiMaster}  (see also~\cite[Appendix B]{Kodama:1985bj})  that, similarly to the  Lorentzian case,  the linearised Einstein operator for metrics \eq{21III17.1} on manifolds as in \eq{7IV17.5} leaves invariant the subspaces of ``scalar'', ``vector'', and ``tensor modes''.

The modes $\ourellnormal=0,1$ require separate attention.
 A detailed analysis of this case, in Lorentzian signature, four spacetime dimensions, and spherical black-hole topology, can be found in~\cite{Dotti2016}.%
\footnote{Compare~\cite[Section~5]{JezierskiWaluk}. The variable $\bf x$ used there is actually awkward for the $\ourellnormal=0$ modes because its definition involves an operator which becomes singular at $r=3m$ precisely for this mode. The calculations there for spherically symmetric perturbations become clear if instead of  ${\bf x}$ one uses $ {\mathcal B}^{-1} {\bf x}=(1-\frac{3m}r){\bf x}$. The operator $\mathcal B$ has been introduced there to obtain simple formulae for all remaining modes.}
Indeed, it is shown there
that the $\ourellnormal=0,1$ modes are, up to a gauge transformation, variations of the mass parameter of the Schwarzschild-anti de Sitter metrics ($\ourellnormal=0$), or variations of the angular-momentum parameter when the Schwarzschild-anti de Sitter metric is viewed as a member of the Kerr-anti de Sitter family of metrics ($\ourellnormal=1$). We show in Appendices~\ref{ss2IX17.1}-\ref{s7VII17.101} how to adapt the arguments of~\cite{Dotti2016} to the cases of interest here.

As is well known, solutions of the linearised Riemannian Einstein equations corresponding to variations of the mass parameter are not in $L^2$,
except if $\mhere$ is given by \eq{9VII17.1} below, which is seen as follows:
 Replacing $t$ in \eq{21III17.1} by a $2\pi$-periodic variable $\varphi$ and $r$ by $\rho$, defined as
\bel{23III18.1}
\rho^2=\frac{4f}{(f'(r_0))^2}\,,
\ee
we find, in all spacetime dimensions $d=n+2$,
\bel{23II18.2}
\zgriem = \rho^2 d\varphi^2+\frac{(f'(r_0))^2}{(f'(r(\rho)))^2}d\rho^2+r(\rho)^2 h_K\,.
\ee
Variations of the mass parameter take the form
\bel{23III18.3}
\frac{d\zgriem}{d\mu}=
\bigg[
	\underbrace{\frac{d(f'(r_0))}{d\mu} \frac{2 f'(r_0)}{(f'(r))^2}}_{\text{I}}
	-\frac{2(f'(r_0))^2}{(f'(r))^3}
	\bigg(
		\underbrace{\frac{\partial f'}{\partial \mu}(r)}_{\text{II}}
		+\underbrace{f''(r)\frac{dr}{d\mu}}_{\text{III}}
	\bigg)
\bigg] d\rho^2
+\underbrace{2r\frac{dr}{d\mu}}_{\text{IV}}h_K\,.
\ee

The part marked I contributes an asymptotically constant term to the norm $|d\zgriem/d\mu|^2_\zgriem$, unless $d(f'(r_0))/d\mu$ vanishes. The term coming from II is of order $r^{-2n-2}$.
For the parts III and IV we consider
\beal{23III18.4}
\\\nonumber
\frac{dr}{d\mu}&=&
\frac{\partial f^{-1}}{\partial \mu}\left(\frac{(f'(r_0))^2\rho^2}{4}\right)
+f^{-1\prime}\left(\frac{(f'(r_0))^2\rho^2}{4}\right)
	\frac{f'(r_0)\rho^2}{2}\frac{d (f'(r_0))}{d\mu}
\\\nonumber
&=&O(r^{-n})+O(r^2)\frac{d(f'(r_0))}{d\mu}\,.
\eea
This implies that the terms in $|d\zgriem/d\mu|^2_\zgriem$ coming from III and IV are of order $r^{-2n-2}$ if $d(f'(r_0))/d\mu=0$ and of order $r^{2}$ otherwise.

Therefore $d\zgriem/d\mhere \in L^2$ if and only if
\bel{9VII17.1+}
\frac{d(\VtwoOrF'(r_0))}{d\mhere}= 0
\,.
\ee
This equation has no solution with $r_0>0$ when $K=0$ or $K=-1$, while if  $K=1$ this leads to
%
\bel{9VII17.1}
r_0=r_c:=\ell    \sqrt{\frac{n-1}{n+1}}
\,,
\qquad
\mhere= \mhere_c:=\frac{n}{n+1}
 \left(\ell
   \sqrt{\frac{n-1}{n+1}}\right)^{n-1}
 \,.
\ee

In Appendix~\ref{ss3IX17.1} we review the Riemannian Kerr-anti de Sitter metrics, and we prove there that variations of angular momentum lead to linearised solutions of the vacuum Einstein equations which are not in $L^2$ either.

\section{Master functions}
 \label{s7VI17.1}

 The master functions  of Ishibashi and Kodama, which we denote by ${\Phi}_{S,\modeindex  }$, ${\Phi}_{V,\modeindex  }$, and ${\Phi}_{T,\modeindex  }$, where the index $\modeindex$ runs over all eigenfunctions of $\Delta_{h_K}$ on $\Nman $ (i.e. $\ourellnormal=\ourellnormal(\modeindex)$),
are solutions of a two-dimensional Schr\"odinger equation
\bel{7IV16.1}
 \Delta_\ztgriem {\Phi}_{i,\modeindex  } - V_{i,\ourellnormal } {\Phi}_{i,\modeindex  }  =0
 \,,
 \quad
  i\in \{S,V,T\}
  \,,
\ee
where
\bel{10IV17.1}
 \ztgriem = \VtwoOrF  dt^2 + \frac{dr^2}{\VtwoOrF}
 \,,
\ee
with $\VtwoOrF=\VtwoOrF(r)$ as in \eq{21III17.1},
while
%
\beal{10IV17.3}
 \lefteqn{
 V_{S, k }
 =
 \frac{1}{16r^2 \left(m+x n(n+1)/2\right)^2}
 \times
}
 &&
\\
 &\Big
 \{\hspace{-.45cm}
 &\nn
 \big[n^3(n+2)(n+1)^2x^2-12n^2(n+1)(n-2)m x \\
 &&\nn \phantom{[} + 4(n-2)(n-4)m^2\big] \lambdaonesquare r^2
\\
 &&\nn +n^4 (n+1)^2x^3\\
 &&\nn +n(n+1)\big[4(2n^2-3n+4)m\\
 &&\nn \phantom{+n(n+1)[}+n(n-2)(n-4)(n+1)K\big]x^2\\
 &&\nn -12n\big[(n-4)m+n(n+1)(n-2)K\big]mx\\
 &&\nn +16m^3+4Kn(n+2)m^2\Big\}
 \,,
 \quad
 k^2 > nK
 \,,
\\
 \lefteqn{
 V_{V,k_V  }
  =
 \frac1{r^2}\left\{k^2_V+K+\frac{n(n-2)}4K\right.
 }
 &&
\\
&&\nn\phantom{xxx}\left.+
\frac{n(n-2)}4\lambdaonesquare r^2-3\frac{n^2\mhere }{2r^{n-1}}\right\}
 \,,
\
 k_V^2 - (n-1)K> 0
 \,,
\\
  \label{11IV17.2}
 \lefteqn{
 V_{T,k_T }
 =
 \frac1{r^2}\left\{k_T^2
 + 2 K +\frac{n(n-2)}4K\right.
}
&&
\\
&&\nn \phantom{xxx}+\left.
\frac{n(n+2)}4\lambdaonesquare r^2+\frac{n^2\mhere }{2r^{n-1}}\right\}
\,,
\eea

where
$$
 x=2\mu r^{1-n}\,,
  \quad
   m=k^2-nK
   \,.
$$
Note that, compared to the equivalent equations in~\cite{KodamaIshibashiMaster}, our equation \eqref{7IV16.1} omits a term $\VtwoOrF(r)^{-1}$ in front of $V_{i,\ourellnormal}$. We have absorbed this factor into the definition of the potentials $V_{i,\ourellnormal}$.  For a sphere there are no master functions for $\ourellnormal=0,1$;   for a torus no master functions $\Phi_{S,I}$ and $\Phi_{V,I}$ for $k=0=k_V$;
and no master functions for $k=0$ (
``scalar master potential'') in the higher genus case.


For $K=1$ we have $k^2=\ourellnormal(\ourellnormal+n-1)$ with $\ourellnormal \ge 0$, $k_V^2=\ourellnormal(\ourellnormal+n-1) - 1$ with $\ourellnormal \ge 1$, and $k_T^2=\ourellnormal(\ourellnormal+n-1)-2$ with $\ourellnormal\ge 2$.

For $K=0$ and a flat torus at infinity,
$$
 \mathbbm{T}^n=\underbrace{
 S^1\times \cdots \times S^1
 }_{\text{$n$ factors}}
 \,,
$$
with each $S^1$-coordinate of period $2\pi$, and with $h_{\kappahere=0} \equiv \gamma_{ij}dx^idx^j$
where $\partial_\mu \gamma_{ij} =0$,
we have
\begin{equation}\label{7IX17,63}
 k^2\,,\ k_V^2\,,\ k_T^2 \in\{\gamma^{ij}k_ik_j\}_{k_i\in \N}
 \,.
\end{equation}
\Eq{7IX17,63} is an immediate consequence of decompositions into Fourier series.

After a periodic identification of $t$ with period $T$ given by \eq{10IV17.1}, the metric \eq{10IV17.1} becomes a smooth rotation-invariant conformally-compactifiable metric on  $\R^2$, with a smooth center of rotation at $\rreg =0$ (equivalently, at $r=\rzero $).

Chasing through the definitions of~\cite{KodamaIshibashiMaster} we show in Appendices \ref{sec:falloff_alln}, \ref{s5VII17.1}, and \ref{s5VII17.7}, using the notation there, that for linearised solutions which are in $L^2$ we have
$H_T=O(\rho^\ourn)$, $|f_a|_\zgriem=O(\rho^\ourn)$, $|f_{ab}|_\zgriem=O(\rho^{ n+1 })$, etc., resulting in
\beal{10IV17.5}
{\Phi}_{S,\modeindex }
 & = &
  \left\{
    \begin{array}{ll}
        O(\rho^{n/2-1}), & \hbox{$n>2$;} \\
      o(1), & \hbox{$n=2$,}
    \end{array}
  \right.
\\
{\Phi}_{V, \modeindex  }
 & = &
  O(\rho^{n/2})
 \,,
\\
{\Phi}_{T,\modeindex  }
 & = &
  O(\rho^{n/2+1})
 \,.
\eea
%

\subsection{Vanishing via the maximum principle}
 \label{ss15X17.1}

We want to prove the vanishing of the master functions for $L^2$-elements of the kernel. The simplest case  occurs when  the potentials in \eq{10IV17.3}-\eq{11IV17.2} are  non-negative.
Since all the $\Phi_{i,\modeindex  }$'s tend to zero  as $r$ tends to infinity,
the vanishing of the corresponding master function  follows from the maximum principle. This leads to rigorous statements for restricted dimensions and masses. For the purposes of the proof of Theorem~\ref{T15X17.1}, the main conclusions of the analysis that follows in this section are:

\begin{proposition}
  \label{P15X17}
Let $h$ be an element of the $L^2$-kernel of the operator $P_L$ defined in \eqref{11IV17.5}. Then:
\begin{enumerate}
  \item The associated tensor master functions vanish.
  \item The associated scalar master functions vanish if $n=2$ and \emph{either} $K=0$ \emph{or}  $K=1$  and  $\ourellnormal \geq 2$.
\end{enumerate}
\end{proposition}

\begin{remark}
\label{R15X17.2}
Note that the above concerns only the master modes, i.e.\ those which are controlled by the master functions ${\Phi}_{i,\modeindex  }$. The scalar modes with $k=0$ (in particular, the modes corresponding to the variation of mass), and the modes $K=1$ and $\ourellnormal=1$ (which include the variation of angular momentum)
 require separate attention.
\qedskip
\end{remark}

Since we have already shown decay of the master functions, it remains to analyze positivity of the potentials occurring in the master equations.
We have not attempted an exhaustive analysis of this,
but we certainly have positivity under the following circumstances,  keeping in mind that we are working in the region where $\VtwoOrF>0$:
\begin{itemize}
\item  for $V_{T,k_T}$.

\item
 When $n=2$ the formula~\eqref{10IV17.3} for $V_{S,k}$ simplifies to
\bel{15IX17.1}
\left\{
\begin{array}{ll}
	V_{S,k}=\frac{1}{3}\left(\frac{6\mu}{r^3}
 +\frac{k^2-2}{r^2}+\frac{2\ell^2(k^6-3k^4+4)+216\mu^2)}{(6\mu+(k^2-2)r)^2}\right),
   & \hbox{$K=1$;}
\\
	V_{S,k}=
 \frac{1}{3}\left(\frac{6\mu}{r^3}
 +\frac{k^2}{r^2}+\frac{2\ell^2 k^6+216\mu^2}{\ell^2(6\mu+k^2r)^2}\right),
  & \hbox{$K=0$;}
\\
	V_{S,k}=
\frac{1}{3}\left(\frac{6\mu}{r^3}
+\frac{k^2+2}{r^2}+\frac{2\ell^2 (k^6+3k^4-4)+216\mu^2}{\ell^2(6\mu+(k^2+2) r)^2}\right),
& \hbox{$K=-1$.}
\end{array}
\right.
\ee
So  $V_{S,k }\ge 0 $ if $n=2$
and \emph{either}
$K=0$,
\emph{or} $K=-1$ and $\mu\ge 0$,
 \emph{or}  $K=1$
and  $\ourellnormal \geq 1$.
Further, when $K=0$ and  $n> 4$
one also finds $V_{S,k}\ge 0$ when $0< \mhere<\mhere(n)\ell^{n-1}$ for a function $\mhere(n)$ which solves a polynomial equation. Indeed, by inserting $\mhere=0$ into \eqref{10IV17.3}, we see that for $n>4$ the scalar potential is positive for small masses. As we also know that it is positive at large $r$, we find that the limiting value of $\mhere$ is reached when $V_{S,k}$ has a minimum with value zero, i.e. when the resultant of the denominators of $V_{S,k}$ and $\partial_r V_{S,k}$, after dividing each by a suitable power of $r$ corresponding to the root $r=0$, has a positive solution. This resultant is a polynomial in $\mhere$ and $k$. Now, $V_{S,k}$ is positive for all $r$ and sufficiently large $k$ at any fixed value of $\mhere$, so the problem is solved by choosing the smallest zero of a finite number of polynomials in $\mhere$ parameterised by $k$.
Similarly when $K=1$ and $n\in \{4,5\}$.

Note that
\bel{10IX17.1}
 V_{S,k }\to_{r\to\infty} V_{S,k }(\infty):=\frac{(n-2)(n-4)}{4\ell^2}
 \,,
\ee
so that the dimensions $n=3,4$ (thus, spacetime dimensions five and six) require different considerations in any case.

   \item
 It holds that $V_{V,k_V}\ge 0 $ if
   %
\bel{10IX17.11}
   \left\{
     \begin{array}{ll}
       \frac{n(n+1)}{2}(2\mhere\ell^{1-n})^{2/(n+1)} \le k_V^2, & \hbox{$K=0$;} \\
       \mbox{$l\geq 2$  and $0<\mhere<\ell$}, & \hbox{$K=1$ and $n=2$;}
       \\
       \mbox{$l\geq 2$  and $0<\mhere<\frac{2}{9}\ell^2$}, & \hbox{$K=1$ and $n=3$;}
       \\
       \mbox{$k_V\ge k_{V,1}(n)$  and $0<\mhere<\mhere_1(k_V,n)\ell^{n-1}$}, & \hbox{$K=1$ and $n\ge 4$;}
       \\
       \mbox{$k_V\ge k_{V,-1}(n)$  and $\mhere_{\min{}}<\mhere<\mhere_{-1}(k_V,n)\ell^{n-1}$}, & \hbox{$K=-1$,}
     \end{array}
   \right.
\ee
for some functions $k_{V,K}(n)$ and $\mhere_K(k_V,n)>0$, where   $\mhere_{\min{}}$ is given by \eqref{18X17.31}.
\end{itemize}
%

\subsection{Vanishing using the bottom of the spectrum}
    \label{ss5VII17.3}

In this section we will prove further vanishing theorems for the master functions by studying the first  $L^2$-eigenvalue $\lambda_1$, and more precisely the kernel, of  Schr\"odinger operators $\nabla^*\nabla+V$ with a smooth potential $V$ and an asymptotically hyperbolic metric $\ztgriem$ on $\R^2$,
$$
\ztgriem=d \rreg ^2 +  \rreg ^2 H(\rreg ) dt^2 =\frac{dr^2}{\VtwoOrF}+\VtwoOrF  dt^2
\, ,
$$
where, as before,
%
$$
 \VtwoOrF=\frac{r^2}{\ell^2}+K-\frac{2\mu}{r^{n-1}}>0
  \,,
  \quad
  \frac{d\VtwoOrF}{dr}=2\left(\frac{r}{\ell^2}+(n-1)\frac{\mu}{r^{n}}\right)>0
 \,,
$$
with $\rreg$ is as in \eqref{21III17.2}.
Recall that $r^{-s}$ is in $L^2$ if and only if  $s>1/2$.
%
%
Using Theorem C of~\cite{Lee:fredholm} we have the following:

\begin{lemma}\label{kernelmasterpotential}
Let  $V$ be a smooth potential on $\R^2$. We assume that  $-\Delta+V$ has a non-trivial indicial interval with smallest characteristic index $s_-$. If $-\Delta+V$  has no $L^2$-kernel, then this last operator
 has no  non trivial function of order $o(r^{-s_-})$ at infinity in its kernel.
\end{lemma}

We would like to apply this lemma to (\ref{7IV16.1}).
We note the following indicial exponents  for the master equations, which turn-out to depend only on the type of the mode: in obvious notation,
\begin{eqnarray}
 \label{15X17.2}
  s^S_{\pm} &=& \frac{1\pm\sqrt{1+(n-2)(n-4)}}2
  \,,
\\
 \label{15X17.3}
 s^V_{\pm}
  & = &
   \frac{1\pm\sqrt{1+n(n-2)}}2
   \,,
\\
 \label{15X17.30}
 s^T_{\pm}
  & = &
   \frac{1\pm\sqrt{1+n(n+2)}}2
    \,.
\end{eqnarray}

\begin{remark}
  \label{R15X17.31}
Note that for $n=3$ we have $s^S_{\pm}=\frac{1}2$ so we can not use directly
Theorem C of~\cite{Lee:fredholm} as in Lemma~\ref{kernelmasterpotential}, but this weight $1/2$ is the critical weight to be in $L^2$
so the conclusion of  Lemma \ref{kernelmasterpotential}
remains true if  {we} replace $o(r^{-s_-})$ with $O(r^{-s_--\epsilon})$ for some $\epsilon>0$, which is the case in our applications.
\qed
\end{remark}

We now study the $L^2$-spectrum of our Schr\"odinger operators.

\begin{lemma}
 \label{L15X17.1}
 Let
 $X$ be a vector field in {\rm $W^{1,1}_{\textrm{loc}}$}.
The first $L^2$-eigenvalue of $\nabla^*\nabla+V$ acting on functions
is bigger or equal than the  almost-everywhere-infimum of
\begin{equation}\label{15X17.1}
 \tilde V :=\nabla_iX^i-|X|^2+V
 \,.
\end{equation}
\end{lemma}

\begin{proof}
Let $\delta_X\in\R\cup\{-\infty\}$ be the a.e.-infimum of $\tilde V$.
For any function $u$ smooth with compact support we have  $\int \nabla_i(u^2X^i)=0$,
so
$$
\int|\nabla u|^2+\int|X|^2u^2\geq -2\int uX^i\nabla_i u=\int u^2\nabla_iX^i.
$$
We thus obtain
$$
\int u(\nabla^*\nabla+V)u\geq\int \tilde V u^2\geq \delta_X\int u^2.
$$
\end{proof}

In order to apply Lemma~\ref{L15X17.1}, we need to find a vector field $X\in W^{1,1}_{\textrm{loc}}$ such that we have, weakly,
$$
 \tilde V=\div X-|X|^2+V\geq 0
 \,,
$$
with $\tilde V$ positive on an open set.
Indeed, if such a vector field $X$ exists,  and if $u$ is in  the $L^2$-kernel of $-\Delta+V$,
then (see the last inequality in the proof above) $u$ has to vanish on the open set  where $\tilde V>0$, so by unique continuation $u=0$ everywhere.

We look for $X$ of the form $X=S\partial_r$ where $S$ is a function on $\R^2$, then
\bel{17X12.01}
 \tilde V=\div X-|X|^2+V=\partial_rS-\frac{S^2}f+V
  \,.
\ee

This should be compared with the  Ishibashi-Kodama modified potential $\tilde V$~\cite{IshibashiKodamaStability}
which differs from ours by  a multiplicative factor $f$
(similarly to the potentials $V_{\ldots}$ entering the master equations).

So whenever their $\tilde V$ is non negative with their choice of $S$, we obtain  positivity of our $\tilde V$ by taking  $X=S \partial_r$. We can then apply Lemma~\ref{L15X17.1}
after verifying that  $X$ belongs to  $W^{1,1}_{\textrm{loc}}$.  This works very well for vector modes in any dimension, leading to:

\begin{corollary}
 \label{C15X17.1}
Under the condition $k_V^2>(n-1)K$,
the $L^2$ kernel of $-\Delta+V_{V,k_V}$ is trivial for all $n\geq2$ and  $K\in\{-1,0,1\}$.
We deduce that any function in the kernel which is of order $o(r^{-s^V_-})$ at infinity, where $s^V_-$ is given by \eq{15X17.3}, is trivial.
\end{corollary}

\begin{Remark}
  \label{R15X17.1}
Corollary~\ref{C15X17.1} applies to $K=1$ and $l\geq2$, since then we have
$$
 k_V^2-(n-1)K=(l-1)(l+n)>0
 \,.
$$
\qedskip
\end{Remark}

\noindent{\sc Proof of Corollary~\ref{C15X17.1}}:
In~\cite[Equation~(2.21)]{IshibashiKodamaStability}
one takes  $S=\frac{nf}{2r}$, so
we take $X=\frac{nf}{2r}\partial_r$, which is in $W^{1,1}_{\textrm{loc}}$.
With this choice of $X$,  the potential in the vector master equation (as  in~\cite[Equation~(2.17)]{IshibashiKodamaStability})
$$
V_V=\frac1{r^2}\left[k_V^2-(n-1)K+\frac{n(n+2)}4f-\frac{n}{2}r\frac{df}{dr}\right],
$$
is transformed to
$$
\tilde V_V=\frac1{r^2}\left[k_V^2-(n-1)K\right]> 0
 \,.
$$
Whence the triviality of the kernel, as claimed.
\qedskip

Let us turn our attention to the scalar potential. The following corollary of Lemma~\ref{L15X17.1} gives an alternative proof of point (2) of Proposition~\ref{P15X17}, and extends that last proposition to the case $K=-1$ (cf., however, Remark~\ref{R15X17.2}):

\begin{corollary}\label{coroTrivialmasterS}
Under the condition $k^2>\max(0,2K)$, the $L^2$ kernel of $-\Delta+V_{S,k}$ is trivial for  $n=2$ for any $K=-1,0,1$.
We deduce that any function of order $o(1)$ at infinity and  in the kernel  is trivial.
\end{corollary}

\begin{proof}
We use the function $S$ chosen in~\cite[Equation~(6.23)]{KodamaIshibashiMasterChargedBH}, where
here we have $n=2$, $Q=\delta=0$, $H=h_-=H_-=m+\frac{6\mu}r$, $m=k^2-2K$, that is
$$
S=\frac f H \frac{dH}{dr},
$$
so $X=S\partial_r$ is in $W^{1,1}_{\textrm{loc}}$.
Our scalar potential $V_S$  becomes (see~\cite[Equation~(6.24)]{KodamaIshibashiMasterChargedBH}
with a factor $f$ removed)
$$
\tilde V_S=\frac{k^2m}{r^2H},
$$
which is positive if $k\neq0$ and  $m=k^2-2K>0$.
\end{proof}


To continue, set
\bel{2X17.1}
 \lambda : = \ell \inf_{r\ge r_0}\frac{ d    \sqrt{\VtwoOrF}}{dr}
 \in
 \left\{
   \begin{array}{ll}
     \{1\}, & \hbox{$K\in \{-1,0\}$;} \\
      (0,1), & \hbox{$K= 1$.}
   \end{array}
 \right.
\ee
We have:

\begin{corollary}
 \label{C15X17.11}
The $L^2$ first eigenvalue of $\nabla^*\nabla$ acting on functions
is bigger or equal than $\lambda^2/(4\ell^2)$.
\end{corollary}

\begin{proof}
We choose $X=c\sqrt f\partial_r\in W^{1,1}_{\textrm{loc}}$, where the positive constant $c$ will  be chosen later. It follows that  $|X|=c$ and
$\div X=c\partial_r(\sqrt f)>0,$ with
$$
(\div X)^2=c^2
\frac{(\partial_r\VtwoOrF)^2}{4\VtwoOrF}=c^2\left(\frac{1}{\ell^2}+(n-1)\frac{\mu}{r^{n+1}}\right)^2
\left(\frac{1}{\ell^2}+K\frac{1}{r^2}-\frac{2\mu}{r^{n+1}}\right)^{-1},
$$
where $2\mu=r_0^{n+1}/{\ell^2}+Kr_0^{n-1}$, $r_0>0$, $r\geq r_0$.
If $K=0,-1$ the function $(\div X)^2$ is decreasing (recall that we have assumed (\ref{18X17.31}) when $K=-1$)
 and tends to $c^2\ell^{-2}$,
but for $K=1$, this function  has a
positive infimum  less than $c^2\ell^{-2}$, say $\lambda^2c^2\ell^{-2}$.
 Thus
$$
\div X-|X|^2\geq \lambda c\ell^{-1}-c^2\
 \,.
$$
The  right-hand side is maximised by $c=\lambda\ell^{-1}/2$, giving
$$
\div X-|X|^2\geq \lambda^2/(4\ell^2).
$$
\nopagebreak
\end{proof}

To apply Corollary~\ref{C15X17.11}, note that  the $L^2$-kernel of $-\Delta+V$ will  be trivial if $V+\lambda^2/(4\ell^2)\geq 0$ and
is strictly positive somewhere.
More generally, from the proof above with $X=\frac{\lambda}{2\ell}\sqrt f\partial_r$, this kernel will be trivial if
\bel{17X2.01}
\tilde V=V+\frac{\lambda}{2\ell}\left(\frac{d\sqrt{f}}{dr}-\frac{\lambda}{2\ell}\right)\geq0,
\ee
and  positive on an open set.
This leads to (compare Proposition~\ref{P15X17} and Remark~\ref{R15X17.2}):%

\begin{proposition}
  \label{P15X17.2}
Let $h$ be an element of the $L^2$-kernel of the operator $P_L$ defined in \eqref{11IV17.5}.
Let $n\in \{4,5\}$ if $K=1$, and $n\in \N$ otherwise.

\begin{enumerate}
  \item  Let $K\in\{ 0,1\}$. There exists a function $\mu(K,n)>0$ such that for $0<\mu < \mu(K,n)$ the corresponding scalar master functions vanish.
  \item Let $K=-1$, and let $\lambda_1>0$ be the first non-zero eigenvalue of the scalar Laplacian of $(\Nman,h_K)$. There exists a function $\mu(\lambda_1,n)> \mu_{\min{}}$ and an $\epsilon(n)>0$  such that for $\lambda_1\geq\epsilon(n)$   and  $\mu_{\min{}}<\mu < \mu(\lambda_1,n)$ the corresponding scalar master functions vanish.
\end{enumerate}
\end{proposition}

\begin{remark}
	By a result of Schoen~\cite{schoen1982} we have, for the case $K=-1$, the lower bound $\lambda_1 \geq (n-1)^2/4$  under the assumption that the volume of ${}^n N$ is sufficiently small, as described in detail there.
We have checked in dimensions $2\le n\le 10 $ that Schoen's estimate is sufficient to obtain $\mu(\lambda_1,n)>0$, we give  approximate values of an upper bound of the $\mu$'s allowed  in Table \ref{t:mulimits}.
Thus  the result is not empty.

The actual values of $\epsilon(n)$ are lower than this estimate, we give approximate values for dimensions $2\le n\le 10 $ in Table \ref{t:epsilon}.
\end{remark}

Similarly to the discussion in the paragraph following \eqref{15IX17.1},
Equation~\eqref{17X2.01} can be used to determine $\mu(K,n)$ as a root of a high order polynomial.
These polynomials grow  with dimension, being extremely long already in $n=4$, so that they cannot be usefully displayed here.

As an illustration, we list approximate numerical values of $\mu(K,n)$ for dimensions $n\leq10$ in Table \ref{t:mulimits} and the corresponding values of $r_0$ in Table \ref{t:r0limits}.
\begin{table}[htbp]
\begin{tabular}{c|ccc}
	$n$ & $K=0$ & $K=1$ & $K=-1$ \\
	\hline
	$2$ & $\infty$ & $\infty$ & $\infty$  \\
	$3$ & $1/48$ & -- & $1/24$ \\
	$4$ & $4.4\times 10^{-3}$ & $3.5\times 10^{-1}$ & $0.11$ \\
	$5$ & $4.7\times 10^{-4}$ & $2.2\times 10^{-2}$ & $0.12$ \\
	$6$ & $3.9\times 10^{-5}$ & -- & $0.19$ \\
	$7$ & $2.7\times 10^{-6}$ & -- & $0.26$ \\
	$8$ & $1.6\times10^{-7}$ & -- & $0.35$ \\
	$9$ & $8.6\times 10^{-9}$ & -- & $0.45$\\
	$10$ & $4.0\times 10^{-10}$ & -- & $0.57$
\end{tabular}
\caption{For $\ell^{n-1}\mu$ between  $\ell^{n-1}\mu_{\min{}}$ (where $\mu_{\min{}}$ is the lowest value giving positive $r_0$, i.e. $\mu_{\min{}}=0$ for $K\in\{0,1\}$ and $\mu_{\min{}}<0$ for $K=-1$, given by \eqref{18X17.31}) and the values shown here we obtain trivial $L^2$ kernel of $-\Delta+V_S$. For $K=-1$ we have used the fact that the first non-zero eigenvalue of the Laplacian on ${}^n N$ is $\geq (n-1)^2/4$, which holds by~\cite{schoen1982} under the assumption that the volume of ${}^n N$ is sufficiently small, cf.~\cite{schoen1982} for details. For the cases marked ``--'' there is no $\mu>0$ such that the inequality \eqref{17X2.01} is satisfied for all relevant $k$ (i.e. those not covered by the linearised Birkhoff theorem or the discussion of $l=1$ modes in appendix \ref{s9IX17.1}).}
\label{t:mulimits}
\end{table}

\begin{table}[htbp]
	\begin{tabular}{c|ccccccccc}
		$n$ & $2$ & $3$ & $4$ & $5$ & $6$ & $7$ & $8$ & $9$ & $10$ \\
		\hline
		$\epsilon(n)$ & $0$ & $0.76$ & $1.13$ & $1.51$ & $1.57$ & $1.88$ & $1.98$ & $2.26$ & $2.33$
	\end{tabular}
	\caption{This table shows approximate values of the function $\epsilon(n)$, appearing in Proposition \ref{P15X17.2}, obtained from the inequality \eqref{17X2.01}.}
	\label{t:epsilon}
\end{table}

\begin{table}[htbp]
	\begin{tabular}{c|ccc}
		$n$ & $K=0$ & $K=1$ & $K=-1$\\
		\hline
		$2$ & $\infty$ & $\infty$ & $\infty$\\
		$3$ & $0.45$ & -- & $1.01$\\
		$4$ & $0.38$ & $0.76$ & $1.08$\\
		$5$ & $0.31$ & $0.43$ & $1.11$\\
		$6$ & $0.25$ & -- & $1.13$\\
		$7$ & $0.22$ & -- & $1.12$\\
		$8$ & $0.19$ & -- & $1.13$\\
		$9$ & $0.16$ & -- & $1.14$\\
		$10$ & $0.14$ & -- & $1.15$\\
	\end{tabular}
	\caption{This table is similar to Table~\ref{t:mulimits}, except that we list the value of $r_0$, the largest zero of $\VtwoOrF$, rather that the value of $\mu$. Thus, for $r_0$ smaller than the values shown here we obtain a trivial $L^2$ kernel of $-\Delta+V_S$ with $\ell=1$.  }
	\label{t:r0limits}
\end{table}

Let us return to the definition \eqref{17X12.01} of $\tilde V$. It is tempting to try to choose $S$ so that $
\tilde V\geq 0$ globally, thus solving the equation
\bel{17X12.01+}
 \partial_rS-\frac{S^2}f+V =\tilde{V}\geq 0\,,
\ee
for some given function $\tilde{V}(r)$ which is positive in an open set.%
\footnote{Once this work was completed we noticed~\cite{Kimura}, where very similar ideas are used in a related context.}
For this, we numerically solved the ODE \eqref{17X12.01} for $S$ using {\sc Mathematica}, with $\tilde{V}_S$ set to zero, with initial value $S(r_0)=0$.
The integration was performed using the \texttt{StiffnessSwitching} method and a \texttt{WorkingPrecision} of 30--100 depending on the value of $\mu$. If the resulting solutions have a zero at a point $r_{S,0}$ such that the potential $V_S$ is positive for all $r\geq r_{S,0}$, we extend $S$ for $r\ge r_{S,0}$  by setting it to be equal to zero there, giving a non-negative $\tilde{V}_S$ for all $r>r_0$ and $\tilde{V}_S>0$ for $r\geq r_{S,0}$. This procedure works for all attempted mass parameters
 and dimensions $n\geq 5$  for $K=0,1$, and all positive masses for $K=-1$.  Indeed, we tested masses in all orders of magnitude between the limits in Table \ref{t:mulimits} and $\mhere=10^{30}$ (with $\ell=1$) for dimensions $n=5,6,100$, and at least 5 values of the mass parameter within this range for each dimension  $5\leq n\leq 100$.

In dimension $n=4$, for $K=1$, this construction of $S$ works for $\mu$ small enough, giving a larger mass range than that obtained from \eqref{17X2.01} (up to $\mu\approx 1.5$ for $\ell=1$). A typical result for $S$ and $\tilde{V}_S$ is shown in Figure~\ref{f:num}.
\begin{figure}[htbp]
	\includegraphics[width=0.8 \textwidth]{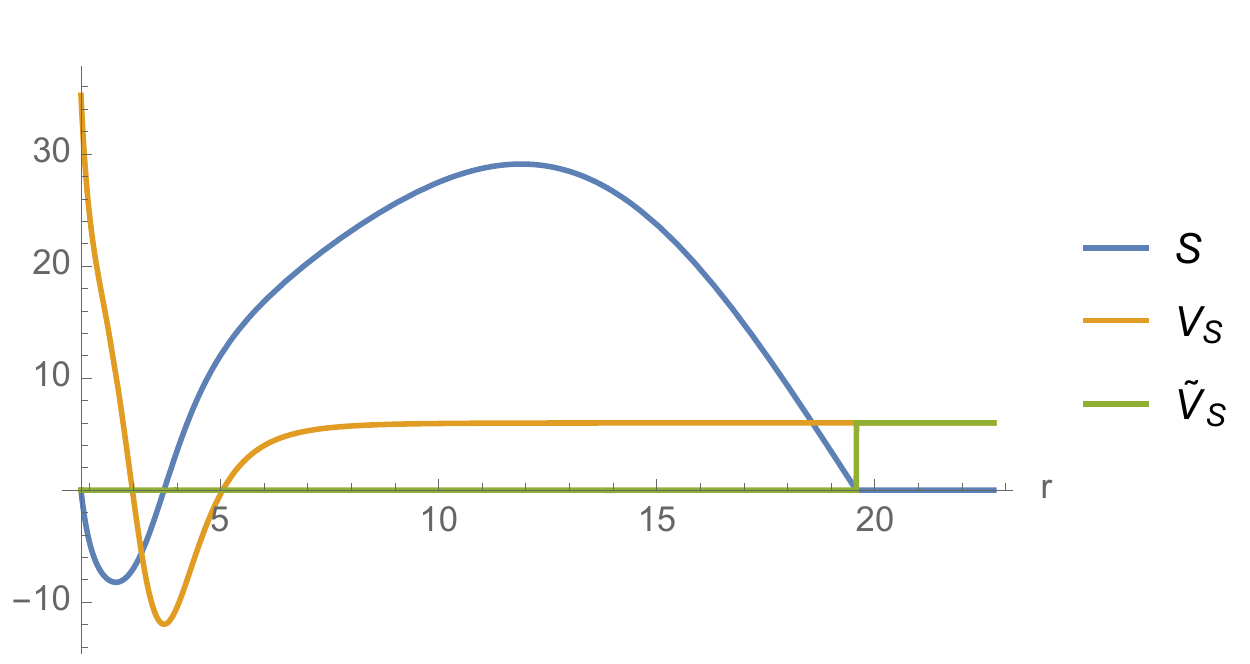}
	\caption{A typical numerical solution for $S$ with $\tilde V=0$, cut off at the third zero of $S$. Here $n=7$, $K=0$, $k=1$, $\mu=100$, $\ell=1$.}
	\label{f:num}
\end{figure}

In dimensions $n=3$ and $4$  for all $K$, and for $K=-1$ with negative mass parameter in all dimensions,  the numerical solution $S$ obtained in this way does not have a zero (except for $K=1$ with low mass) but appears to exist globally. Since we need positivity of $\tilde V$ somewhere to conclude, we instead numerically solved the ODE \eqref{17X12.01} with $\tilde{V}=\epsilon r^{-2}$ for $\epsilon>0$ (for $\mu>0$ we can choose $\epsilon=1$, for $\mu<0$ the choice of $\epsilon$ depends on $\mu$), with again initial value $S(r_0)=0$.
The resulting numerical solution $S$ appears to grow  asymptotically linearly for large $r$ and we therefore cannot set it to zero at some $r_{S,0}$, as a negative jump in $S$ would add a negative distributional component to $\tilde V$. Instead we just use the solution directly, without cutting off at finite distance. This works in fact for all masses,  dimensions, and $K=0,1,-1$, but the results are less conclusive as one has then to rely on the behaviour of the solution up to $r=\infty$, while the numerical integration necessarily stops at some finite $r$. A typical result for $S$ and $\tilde{V}_S$ for this approach is shown in Figure~\ref{f:num_2}.
\begin{figure}[htbp]
	\includegraphics[width=0.8 \textwidth]{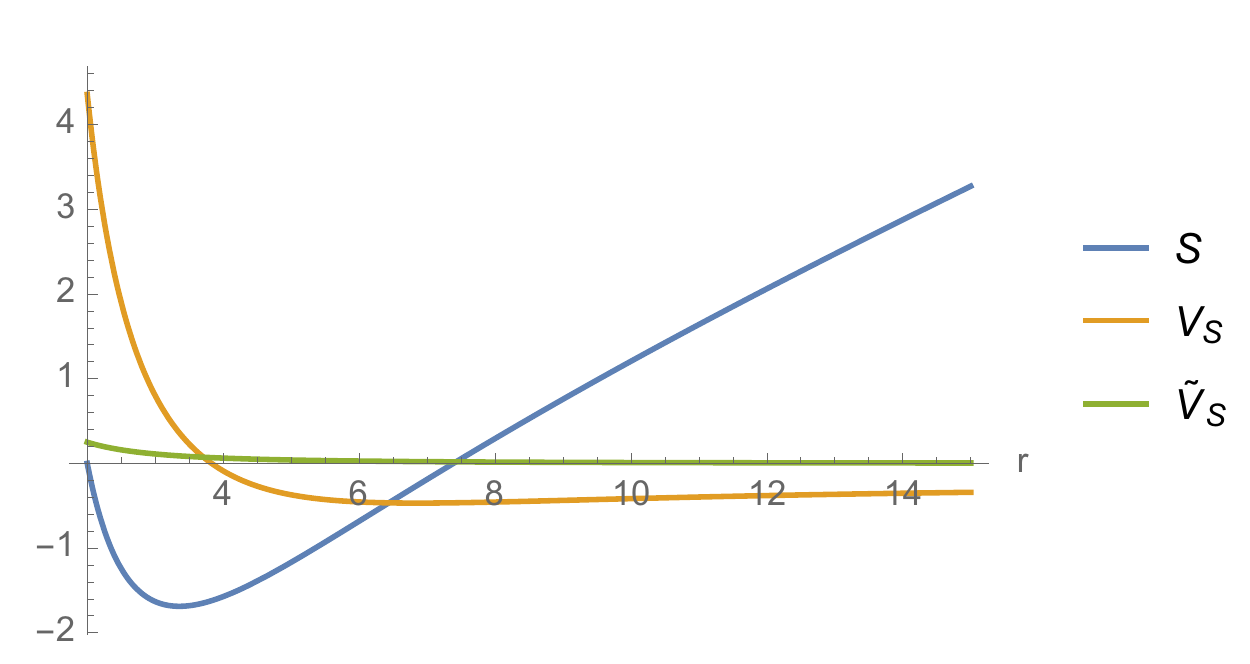}
	\caption{A typical numerical solution for $S$ with $\tilde V=1/r^2$. Here $n=3$, $K=1$, $\ourellnormal=2$, $\mu=10$, $\ell=1$.}
	\label{f:num_2}
\end{figure}

In any case, our numerical experiments strongly hint  at the following conjecture:

\begin{conjecture}
  \label{C15X17.21}
There are no non-trivial scalar master modes associated with the $L^2$-kernel of the operator $P_L$ given by \eq{11IV17.5}.
\end{conjecture}


%


\section{Gauge vectors for vanishing master functions}
 \label{s5VII17.2}

In this section we study the consequences of the vanishing of the master functions.
We consider perturbations of the $(n+2)$-dimensional metric
\[
\zgriem=g_{ab} dx^a dx^b + r^2\gamma_{ij} dx^i dx^j\,,
\]
where
\[
g_{ab} dx^a dx^b=\VtwoOrF dt^2 + \VtwoOrF^{-1}dr^2\,,
\]
and where $\gamma$ is the metric of an $n$-dimensional manifold with constant sectional curvature $K$. The indices $a,b$ take values $t, r$, while $i, j$ are ``angular'' indices. We will denote by $D$ the covariant derivative associated to $g$ and by $\hat{D}$ that of $\gamma$.

As emphasised in~\cite{KodamaIshibashiMaster}, a general metric perturbation $\deltazgriem $ can be split into ``scalar'', ``vector'', and ``tensor'' parts as
\bel{7VI17.101}
 \deltazgriem =\deltazgriem ^S+\deltazgriem ^V + \deltazgriem ^T
 \,.
\ee
An analytic description of the splitting \eq{7VI17.101}, without referring to a mode-decomposition, is presented in  Appendix~\ref{appendicedecompo}.

We show in Appendix~\ref{s5VII17.21} below that the $TT$ conditions are consistent with the splitting above.

The components in \eq{7VI17.101} can be expanded as~\cite[Sections 2.1, 5.1 and 5.2]{KodamaIshibashiMaster}
\begin{align}
	\label{06IX17.1}
	\deltazgriem ^S_{ab}&=\sum_{\efindex} \fabSI \mathbb{S}^{\efindex}\,,
	&\deltazgriem ^S_{ai}&=\sum_{\efindex} r \faSI \mathbb{S}^{\efindex}_i\,,
	&\deltazgriem ^S_{ij}&=\sum_{\efindex} 2r^2(\HLSI\gamma_{ij}\mathbb{S}^{\efindex}+\HTSI\mathbb{S}^{\efindex}_{ij})
 \,,
\\
	\label{06IX17.2}
	\deltazgriem ^V_{ab}&=0\,,
	&\deltazgriem ^V_{ai}&=\sum_{\efindex} r \faVI \mathbb{V}^{\efindex}_i\,,
	&\deltazgriem ^V_{ij}&=\sum_{\efindex} 2r^2 \HTVI
        \mathbb{V}^{\efindex}_{ij}
        \,,
\\
	\label{06IX17.3}
	\deltazgriem ^T_{ab}&=0\,,
	&\deltazgriem ^T_{ai}&=0\,,
	&\deltazgriem ^T_{ij}&=\sum_{\efindex} 2r^2 \HTTI
        \mathbb{T}^{\efindex}_{ij}\,,
\end{align}
where the following holds: for $K=0$, the above are obvious decompositions in Fourier series, with coordinate-independent tensor components leading to zero eigenvalues. For $K=1$
(see, e.g.~\cite{Boucetta1999,Rubin1984})
\begin{align}
(\hat{\Delta}_n+k^2)\mathbb{S}^{\efindex}&=0\,,&& k^2=\ourellnormal(\ourellnormal+n-1)\,, &&\ourellnormal=0,1,2,\dots\,,
\\
 \label{7IX17.62}
 (\hat{\Delta}_n+k_V^2)\mathbb{V}^{\efindex}_i&=0\,,&& k_V^2=\ourellnormal(\ourellnormal+n-1)-1\,, &&\ourellnormal=1,2,\dots\,,\\
(\hat{\Delta}_n+k_T^2)\mathbb{T}^{\efindex}_{ij}&=0\,,&& k_T^2=\ourellnormal(\ourellnormal+n-1)-2\,, &&\ourellnormal=2,\dots\,,
\quad n>2
 \,.
\end{align}

For $K=-1$ there is no example
of compact quotient of $\H^n$ with explicit values of the whole spectrum, but
we have a countable set of increasing  eigenvalues  of finite multiplicity:
$$
k^2=(\lambda_0=0), \lambda_1, \dots
 \,,
  \qquad
k_V^2=\lambda_{V,0}, \dots
 \,,
$$
where $\lambda_{V,0}\geq (n-1)$ by non-negativity of the Hodge Laplacian (\ref{DeltaH}). 
The first non zero eigenvalue $\lambda_1$ is greater
or equal than a computable constant~\cite{schoen1982} that can be chosen to be $(n-1)^2/4$ for sufficiently small volumes.

Whatever the value of $K\in\{-1,0,1\}$ one sets
\begin{align}
\mathbb{S}^{\efindex}_i&=-\frac{1}{k}\hat{D}_i\mathbb{S}^{\efindex}
 \,,
  \quad
   k\ne 0
   \,,
    \label{4VII17.26}
\\
 \label{7IX17.61}
	\mathbb{S}^{\efindex}_{ij}
 &=
 \frac{1}{k^2}\hat{D}_i\hat{D}_j\mathbb{S}^{\efindex}+\frac{1}{n}\gamma_{ij}\mathbb{S}^{\efindex}
  \,,
  \quad
   k \ne 0
   \,,
   \\
 \mathbb{V}^{\efindex}_{ij}&=-\frac{1}{2k_V}(\hat{D}_i\mathbb{V}^{\efindex}_j
  +\hat{D}_j\mathbb{V}^{\efindex}_i)
 =-\frac{1}{2k_V}\mathcal{L}_{\mathbb{V}^{\efindex}}\gamma_{ij}
 \,,
 \quad
  k_V \ne 0
 \,.
\end{align}
We define the sets $\mathbb{\efindex}^\star_S$ and $\IVstar$ of indices $\efindex$ corresponding to modes governed by the scalar and vector master equations:
\beal{17X04.1}
\ISstar&=&\{\efindex \,|\, k(\efindex)>0,\, k^2(\efindex)>n\kappahere \}\,,\\ \label{17X04.2}
\IVstar&=&\{\efindex \,|\, k_V(\efindex)>0,\, k^2_V(\efindex)>(n-1)\kappahere \}\,,
\eea
and denote by $\ISdiamond$ and $\IVdiamond$ their complements.

Note that $
 k_V $ never vanishes when $\kappahere\in\{\pm1\}$, and that those vector fields  $\mathbb{V}^{\efindex}_j$ for which $\ourellnormal(I)$ in \eqref{7IX17.62} equals one, are Killing vector fields of $\gamma_{ij}$ (cf., e.g., the eigenvalue $\lambda^1_1$ for the Hodge Laplacian in~\cite[Theorem 3.1]{Boucetta1999}), so that the associated tensor $ \mathbb{V}^{\efindex}_{ij} $ vanishes.

From now on we assume that the master functions $\Phi_S$, $\Phi_V$, and $\Phi_T$ vanish.

The vanishing of the tensor potential directly gives $\deltazgriem ^T=0$~\cite[(5.4)]{KodamaIshibashiMaster}.

The vanishing of the vector potential implies~\cite[(5.10)-(5.13)]{KodamaIshibashiMaster},
for modes such that $k_V\ne 0$,
%
\begin{equation}\label{4VII17.20}
	\faVI =-\frac{r}{k_V}D_a \HTVI \,.
\end{equation}
Let us define the vector field $\YVstar{}$ as
\begin{equation}\label{8VIII18.7+}
  \YVstar{}:= -r^2\sum_{\efindex \in \IVstar} \frac{\HTVI}{k_V(I)}
 \mathbb{V}^{\efindex}_j \gamma^{ij} \partial_i
  \,.
\end{equation}
The convergence of the series is justified in Appendix~\ref{s5VII17.7}.

Let us denote by $\hVstar{ij}$ and $\hVstar{ia}$ tensors in which we have collected all those modes in $\deltazgriem ^V_{ij}$ and  $\deltazgriem ^V_{ia}$ which are governed by the master equations, and by $\hVdiamond{ij}$ and $\hVdiamond{ia}$ whatever remains:
\beal{7IX18.51}
	 \hVstar{ai}
 &=&
  \sum_{\efindex \in \IVstar} r \faVI \mathbb{V}^{\efindex}_i\,,
\\
	 \hVstar{ij}
 &=&
 \sum_{\efindex \in \IVstar} 2r^2 \HTVI
        \mathbb{V}^{\efindex}_{ij}
\\
	 \hVdiamond{ai}
 &=&
 \sum_{\efindex \in \IVdiamond } r \faVI \mathbb{V}^{\efindex}_i\,,
\\
	 \hVdiamond{ij}
 &=&
 \sum_{\efindex \in \IVdiamond} 2r^2 \HTVI
        \mathbb{V}^{\efindex}_{ij}
\eea
From \eq{8VIII18.06} and \eq{8VIII18.8+}, Appendix~\ref{s5VII17.7}, we  obtain
\bel{8VIII18.6}
 \| D_a (r^{-2} \hVstar{ij}) dx^i dx^j  \|_{H^k(\Nman)}
 \le 2 r^{-2}  \|  \hVstar{ai} dx^i  \|_{H^{k+1}(\Nman)}
  \,,
\ee
\begin{equation}\label{4VII17.21}
\deltazgriem ^V_{ij}=
 \hVs{ij} + r^2 (\mathcal{L}_{\YVstar{}}\gamma)_{ij}
  =
  \hVs{ij} +
  (\mathcal{L}_{\YVstar{}}\zgriem)_{ij}
   \,.
\end{equation}

Using \eqref{4VII17.20}, the mixed components $\deltazgriem ^V_{ai}$ take the form
\begin{equation}
	\deltazgriem ^V_{ai}=\sum_{\efindex} r \faVI  \mathbb V^{\efindex}_i
 =
  \hVs{ai}-\sum_{\efindex \in \IVstar}\frac{r^2}{k_V}  \mathbb V^{\efindex}_i D_a \HTVI
  =\hVs{ai} + D_a \YVstar{i}
   =\hVs{ai} +
   (\mathcal{L}_{\YVstar{}}\zgriem)_{ai}
\end{equation}
(note that $\zgriem_{ij}\gamma^{jk}=r^2\delta^k_i$).
We therefore have
\begin{equation}\label{4VII17.22}
	\deltazgriem ^V=\hVs{} + \mathcal{L}_{\YVstar{}}\zgriem\,.
\end{equation}

We continue with the scalar variations. We define
\begin{equation}\label{7IX17.102}
  	\hSs{ab}=\sum_{I \notin \ISstar } \fabSI \mathbb{S}^{\efindex}\,,
  \quad
  	\hSstar{ab}=\sum_{\efindex \in \ISstar} \fabSI \mathbb{S}^{\efindex}\,,
\end{equation}
with obvious similar definitions for $\hSstar{ai}$, etc.

The vanishing of the scalar master function implies~\cite[(2.7), (2.21) and (3.9)]{KodamaIshibashiMaster}
\begin{align}
\HLSI +\frac{1}{n}\HTSI +\frac{1}{r}D^a r \XaI =0\,,\label{4VII17.23}\\
0=\fabSI+(\mathcal{L}_{X_{\efindex}} g)_{ab} \label{4VII17.24}
\end{align}
where  we have assumed that $k\ne 0$ and
\begin{equation}\label{4VII17.25}
\XaI =\frac{r}{k}(\faSI+\frac{r}{k}D_a \HTSI )\,.
\end{equation}
Using \eqref{4VII17.23}, the angular part of the variation can be written as
\begin{equation}\begin{split}
 \hSstar{ij}
  &=\sum_{\efindex \in \ISstar}2 r^2 \left(\frac{\HTSI }{k^2} \hat{D}_i\hat{D}_j \mathbb{S}^{\efindex}-\gamma_{ij}\mathbb{S}^{\efindex}\frac{1}{r}X(r)\right)
\\
 &=r^2 \mathcal{L}_{\YSstar{}}\gamma_{ij}+2r\YSstar{}(r)\gamma_{ij}
\\
 &=\mathcal{L}_{\YSstar{}}(r^2\gamma_{ij})
 \,,
\end{split}
\end{equation}
where
\begin{equation}
	{\YSstar{}}:=\sum_{\efindex \in \ISstar}\left(\frac{1}{k^2}\HTSI \gamma^{ij}\hat{D}_j\mathbb{S}^{\efindex}\partial_i-\mathbb{S}^{\efindex}X^a \partial_a\right)\,.
\end{equation}
The convergence of all series above can be justified in a way very similar to that of Appendix~\ref{s5VII17.7}, and will be omitted.

Using \eqref{4VII17.25} and \eqref{4VII17.26} the mixed part is given by
\begin{equation}\begin{split}
\deltazgriem ^S_{ai}&=\sum_{\efindex \in \ISstar}\hat{D}_i\mathbb{S}^{\efindex}\left(-\XaI +\frac{r^2}{k^2}D_a \HTSI \right)\\
&=\hat{D}_i Y_{S,a}+ D_a Y_{S,i}=(\mathcal{L}_{\YSstar{}}\zgriem)_{ai}\,.
\end{split}\end{equation}

The $r,t$ part of the variation is finally, using \eqref{4VII17.24},
\begin{equation}
	\deltazgriem ^S_{ab}
	=\sum_{\efindex \in \ISstar}\fabSI\mathbb{S}^{\efindex}
	=-\sum_{\efindex \in \ISstar}\mathbb{S}^{\efindex}(\mathcal{L}_{\YSstar{}} g)_{ab}
	=(\mathcal{L}_{\YSstar{}} g)_{ab}\,.
\end{equation}

In conclusion, we have
\begin{equation}\label{4VII17.27}
 \deltazgriem =\hdiamond{}+\mathcal{L}_{\Ystar{}}\zgriem\,,
\end{equation}
with
\begin{equation}\label{4VII17.28}
 \Ystar{}={\YSstar{}}+\YVstar{}
 =\sum_{\efindex \in \ISstar}\left( \frac{1}{k^2}\HTSI \hat{D}^i\mathbb{S}^{\efindex}\,\partial_i
 -\mathbb{S}^{\efindex}X^a_{\efindex} \partial_a\right)
 -\sum_{\efindex \in \IVstar} \frac{\HTVI }{k_V}\mathbb{V}^{\efindex}
 \,.
\end{equation}

Using the estimates in Appendix \ref{sec:falloff_alln}
  we obtain for the asymptotics of ${\Ystar{}}$
\begin{equation}\label{6IX17.1}
\begin{split}
|{\Ystar{}}|^2&=\zgriem_{tt}(({\Ystar{}})^t)^2 + \zgriem_{rr} (({\Ystar{}})^r)^2+r^2 |({\Ystar{}})^i|^2_\gamma\\ &=O(r^{-2-2n+4})+O(r^{2-2-2n})+O(r^{2-2-2n})\\
&=O(r^{2-2n})\,.
\end{split}
\end{equation}

\section{Triviality of the kernel}
 \label{s8IX17.1}

We are ready now to pass to the proof of non-degeneracy:

\begin{theorem}
 \label{T8IX17.1}
Consider an $(n+2)$-dimensional Riemannian Kottler metric $\zgriem$ as in \eqref{21III17.1}, with $\mu>0$ for $K\in\{0,1\}$ or $\mu$ satisfying \eqref{18X17.31} for $K=-1$, and an axis of rotation at $r=r_0>0$, as described in Section~\ref{sA7IV17.1}. Suppose that

\begin{enumerate}
  \item $K=0$, $n\ge 2 $,
   $\mhere$  as in Proposition~\ref{P15X17.2}; or
  \item $K=1$, $n= 2$,
   $\mhere\not = \mu_c$ given by \eqref{9VII17.1}; or
  \item $K=-1$, $n= 2$.
\end{enumerate}

Then $(\R^2\times \Nman,\zgriem)$ is non-degenerate in the sense described in the introduction.
\end{theorem}

\proof
Let $h\in L^2$ be in the kernel of the operator $P_L$.
Consider the decomposition of $h$ into master scalar, vectorial, and tensor modes, together with their non-master counterparts:
\bel{8IX17.1}
 h =
  \underbrace{
  \hSstar{} + \hSdiamond{}
  }_{h^S}
  +
  \underbrace{
  \hVstar{} + \hVdiamond{}
  }_{h^V}
  +
  \underbrace{
  \hTstar{}
  }_{ h^T}
 \,,
\ee
(note that  all tensorial modes are controlled by master functions).

Suppose, first, that $K=0$.
We show in Appendix~\ref{s7VII17.101} below that
all angle-independent modes $\hdiamond{}$ of $h$ (in particular, all non-master modes)
 are pure gauge:
%
\begin{equation}\label{07IX17.11+}
\hdiamond{}:= \hSdiamond{}+\hVdiamond{}
= \mcL_{\Ydiamond{}} \zgriem
\,,
\quad
| \Ydiamond{} |_\zgriem=O(r^{-3})
\,.
\end{equation}
It follows from Proposition~\ref{P15X17.2} together with the analysis in Section~\ref{s5VII17.1}  that,  for $n\ge 2$,
all master modes are likewise pure gauge for a nontrivial range of mass parameters $\mhere$ (for all $\mhere >0 $ when $n=2$):
\begin{equation}\label{07IX17.11+star}
\hstar{}:= \hSstar{}+\hVstar{}+\hTstar{}= \mcL_{\Ystar{}} \zgriem
\,,
%
\quad
 | \Ystar{}  | _\zgriem=O(r^{1-n})
 \,.
\end{equation}
Thus
\begin{equation}\label{p9XI17.1}
  h = \mcL_{Y} \zgriem
   \,,
   \quad
   Y:= \Ystar{}+\Ydiamond{}\,,
\quad
| Y |_\zgriem=O(r^{-1})
   \,.
\end{equation}

Now, $\deltazgriem$ in \eqref{4VII17.27} is in $TT$-gauge, and the operator obtained by composing the divergence and the trace-free part of the Lie derivative is precisely the operator $L^*L$ considered in~\cite[Proposition~G]{Lee:fredholm}. Keeping in mind that our $n$ is shifted by one as compared to the parameter $n$ used in~\cite{Lee:fredholm}, from~\cite[Proposition~G]{Lee:fredholm} we find that the indicial  radius of
$L^*L$ is $R=\frac{n+3}2$.   From~\cite[Theorem~C(c)]{Lee:fredholm},
choosing $\delta$ there so that the  Sobolev space $C^{k,\alpha}_\delta$ contains only decaying fields, we see that any element $Y$ in the  $L^2$-kernel of $L^*L$ has to decay at infinity at a rate as close to the lower end of the indicial interval as desired.
(In fact a more careful analysis shows that an element of the $L^2$-kernel must decay as
\begin{equation}\label{p9XI17.2}
   |Y|_\zgriem = O(r^{-(\frac{n+1}2+R)})
   = O(r^{-2})
   \,.)
\end{equation}
Since the $L^2$-kernel of $L^*L$ is the same as the
$L^2$-kernel of $L$,
one can invoke~\cite[Proposition~6.2.2]{AndChDiss} to conclude that the $L^2$-kernel of $L^*L$ is always trivial.

Next, \eqref{p9XI17.1} shows that for our gauge field  $Y$ it holds that
  $Y\rho^\delta$ is in $L^2$ for   $\delta >\frac{3-n}2$;
choosing $\delta \in (\frac{3-n}2,\frac{n+3}2)$ we can thus use~\cite [Theorem~C(b)]{Lee:fredholm}
to conclude that $Y\equiv 0$. Hence
\begin{equation}\label{6IX17.20}
  \deltazgriem \equiv 0
  \,,
\end{equation}
which concludes the proof when $\kappahere=0$.

The argument for $\kappahere=\pm 1$ is very similar: when $n=2$
there are no non-master tensor modes; the fact that non-master scalar modes are pure gauge is the contents of the linearised Birkhoff theorems of Appendices~\ref{ss2IX17.1} and \ref{s03X17.1}, while the pure-gauge character of the $\ourellnormal=1$, $K=1$
 modes is established in Appendix~\ref{s9IX17.1}.
\qed
%


%

\appendix

\section{Divergence of a symmetric two-tensor in coordinates}
 \label{s5VII17.21}

The object of this appendix is to show that the $TT$ condition does not mix the scalar, vector, and tensor modes.

Let us consider a warped product metric of the  following form (note that this allows for metrics more general than the metric $\zgriem$ of the main body of the paper):
$$
\mathcal G=g_{ab}(y)dy^ady^b+r^2(y)\gamma_{ij}(x)dx^idx^j,
$$
where the $y^a$'s are local coordinates on a $m$-dimensional  manifold and the $x^i$'s are local coordinates in $n$-dimensions.
The non trivial Christoffel symbols of $\mathcal G$ are:
$$
\Gamma^c_{ab}=\Gamma^c_{ab}(g)\;,\;\;\Gamma^k_{ij}=\Gamma^k_{ij}(\gamma)\;,
$$
$$
\Gamma^c_{ij}=-g^{cb}r\partial_br\,\gamma_{ij}\;,\;\;\Gamma^k_{ib}=r^{-1}\partial_br\,\delta^k_i\;.
$$
The divergence of a 2-tensor is
$$
\nabla_\alpha T^{\alpha\beta}=\partial_\alpha T^{\alpha\beta}+\Gamma^\alpha_{\alpha\sigma}T^{\sigma\beta}+\Gamma^\beta_{\alpha\sigma}T^{\alpha\sigma},$$
or equivalently
$$
\nabla_\alpha T^{\alpha\beta}=
(\sqrt{|\mathcal G|})^{-1}\partial_\alpha(\sqrt{|\mathcal G|}T^{\alpha\beta})+\Gamma^\beta_{\alpha\sigma}T^{\alpha\sigma},
$$
where
$$
\sqrt{|\mathcal G|}=\sqrt{|g|}\sqrt{|\gamma|}r^n.
$$
We deduce
$$
\nabla_\alpha T^{\alpha\beta}=
(\sqrt{|g|})^{-1}r^{-n}\partial_a(r^n\sqrt{|g|}T^{a \beta})+
(\sqrt{|\gamma|})^{-1}\partial_i(\sqrt{|\gamma|}T^{i\beta})+\Gamma^\beta_{\alpha\sigma}T^{\alpha\sigma}.
$$
In particular
$$
\nabla_\alpha T^{\alpha b}=
(\sqrt{|g|})^{-1}r^{-n}\partial_a(r^n\sqrt{|g|}T^{a b})+
(\sqrt{|\gamma|})^{-1}\partial_i(\sqrt{|\gamma|}T^{i b})+\Gamma^b_{ac}T^{ac}
-g^{bc}r\partial_cr\gamma_{ij}T^{ij},
$$
so
\begin{equation}\label{divTb}
\nabla_\alpha T^{\alpha b}=
r^{-n}D_a(r^nT^{a b})+
\hat D _i(T^{i b})
-g^{bc}r\partial_cr\gamma_{ij}T^{ij},
\end{equation}
where $\hat D $ is the $\gamma$-connection and $ D  $ is the $g$-connection.
Similarly
$$
\nabla_\alpha T^{\alpha j}=
(\sqrt{|g|})^{-1}r^{-n}\partial_a(r^n\sqrt{|g|}T^{a j})+
(\sqrt{|\gamma|})^{-1}\partial_i(\sqrt{|\gamma|}T^{i j})+\Gamma^j_{kl}T^{kl}
+r^{-1}\partial_brT^{bj},
$$
so
\begin{equation}\label{divTi}
\nabla_\alpha T^{\alpha j}=
r^{-n} D  _a(r^nT^{a j})+
\hat D _i(T^{i j})
+r^{-1}\partial_brT^{bj}.
\end{equation}

\subsection{Divergence of a scalar variation}
 \label{ss9VII17.1}

Let us write,  as in~\cite[Equation~(2.4)]{KodamaIshibashiMaster},
$$
T^{ab}=\mathbb{S}f^{ab}\,, \ T^{ai}=r^{-1}f^a\mathbb{S}^i\,, \ T^{ij}=2r^{-2}(H_L\mathbb{S}\gamma^{ij}+H_T\mathbb{S}^{ij})\;,
$$
Note that here $\mathbb{S}^i=\gamma^{ij}\mathbb{S}_j$ but $T^{ai}=\zgriem^{ij}T^a{}_j$. From (\ref{divTb}) and (\ref{divTi})
\[\begin{split}
\nabla_\alpha T^{\alpha a}&=\mathbb{S}r^{-n} D  _b(r^nf^{ab})+r^{-1}f^a\hat D _i\mathbb{S}^i
-g^{ac}\partial_cr\gamma_{ij}2r^{-1}(H_L\mathbb{S}\gamma^{ij}+H_T\mathbb{S}^{ij})\\
\nabla_\alpha T^{\alpha j}&=r^{-n} D  _b(r^{n-1}f^{b})\mathbb{S}^j+2r^{-2}
(H_L\hat D ^j\mathbb{S}+H_T\hat D _i\mathbb{S}^{ij})+r^{-2}\partial_brf^b\mathbb{S}^j\\
\end{split}\]
Assume  moreover  that (see~\cite{KodamaIshibashiMaster},
Equations (2.2) and (2.6))
$$
\Delta_\gamma \mathbb{S}=-k^2\mathbb{S}\;,\;\;\mathbb{S}_i=-\frac1 k \partial_i\mathbb{S}\;,\;\;\mathbb{S}_{ij}=\frac1{k^2}\hat D _i
\partial_j \mathbb{S}+\frac1n\gamma_{ij}\mathbb{S},
$$
We obtain
\[\begin{split}
 \nabla_\alpha T^{\alpha a}&=\mathbb{S}r^{-n} D  _b(r^nf^{ab})+r^{-1}f^a k \mathbb{S}
 -g^{ac}\partial_cr\gamma_{ij}2r^{-1}(H_L\mathbb{S}\gamma^{ij}
  +H_T\mathbb{S}^{ij})
 \,,
\\
 &=\mathbb{S}[r^{-n} D  _b(r^nf^{ab})+r^{-1}f^ak-2r^{-1}g^{ac}\partial_cr\,nH_L]
 \,,
\\
 \nabla_\alpha T^{\alpha j}&=r^{-n} D  _b(r^{n-1}f^{b})\mathbb{S}^j+2r^{-2}
 (H_L\hat D ^j\mathbb{S}+H_T\hat D _i\mathbb{S}^{ij})+r^{-2}\partial_brf^b\mathbb{S}^j
 \,,
\\
 &=\hat D ^j\mathbb{S}\bigg[-\frac1kr^{-n} D  _b(r^{n-1}f^{b})+2r^{-2}\left(H_L+H_T(-1+\frac{(n-1)\kappahere }{k^{2}}+\frac1n)\right)\\
 &\phantom{=\hat D^j\mathbb{S}\bigg[}-\frac1kr^{-2}\partial_brf^b\bigg]
 \,.
\end{split}
\]

\subsection{Divergence of a vector variation}
If, we have 
$$
T^{ab}=0\,, \ T^{ai}=r^{-1}f^a\mathbb{V}^i\,, \ T^{ij}=2r^{-2}H_T\mathbb{V}^{ij}\;,
$$
then from (\ref{divTb}) and (\ref{divTi})
\[
 \begin{split}
\nabla_\alpha T^{\alpha a}&=r^{-1}f^a\hat D _i\mathbb{V}^i
-g^{ac}\partial_cr\gamma_{ij}2r^{-1}H_T\mathbb{V}^{ij}
 \,,
\\
\nabla_\alpha T^{\alpha j}&=r^{-n} D  _b(r^{n-1}f^{b})\mathbb{V}^j+2r^{-2}
H_T\hat D _i\mathbb{V}^{ij}+r^{-2}\partial_brf^b\mathbb{V}^j
 \,.
 \end{split}
\]
Assume  moreover  that (see~\cite{KodamaIshibashiMaster}, Equations (5.7a), (5.7b) and (5.9))
$$
\Delta_\gamma \mathbb{V}=-k_V^2\mathbb{V}\;,\;\;\mathbb{V}_{ij}=-\frac1{2k_V}(\hat D _i
\mathbb{V}_j +\hat D _j
\mathbb{V}_i )\;,\;\;\hat D _i\mathbb{V}^i=0
 \,,
$$
where $k_V\ne 0$,
then
\[\begin{split}
\nabla_\alpha T^{\alpha a}&=0
 \,,\\
\nabla_\alpha T^{\alpha j}&=\mathbb{V}^j\left(r^{-n} D  _b(r^{n-1}f^{b})-\frac1{k_V}r^{-2}
H_T[-k_V^2+(n-1)\kappahere ]+r^{-2}\partial_brf^b\right)
 \,.
\end{split}\]

\subsection{Divergence of a tensor variation}
If we assume that
$$
T^{ab}=0\,, \ T^{ai}=0\,, \ T^{ij}=2r^{-2}H_T\mathbb{T}^{ij}\;,
$$
then from (\ref{divTb}) and (\ref{divTi})
\[\begin{split}
\nabla_\alpha T^{\alpha a}&=
-g^{ac}\partial_cr\gamma_{ij}2r^{-1}H_T\mathbb{T}^{ij}
 \,,
\\
\nabla_\alpha T^{\alpha j}&=2r^{-2}
H_T\hat D _i\mathbb{T}^{ij}
 \,.
\end{split}\]
Assume  moreover  that (see~\cite{KodamaIshibashiMaster} equation (5.1a) and (5.1b))
$$
\gamma_{ij}\mathbb{T}^{ij}=0\;,\;\;\hat D _i\mathbb{T}^{ij}=0
\,,
$$
then
\[
\nabla_\alpha T^{\alpha \mu} =0
  \,.
\]

\section{Divergence \& double divergence of $\deltagriem^S_{ij}$ and $\deltagriem^S_{aj}$}
 \label{sA17VIII1}

The aim of this appendix is to derive some divergence identities, as implicitly used in Appendix~\ref{appendicedecompo}.

We start by calculating the divergence of $\mathbb{S}^{\efindex}_{jk}$. Assuming $k\ne 0$ it holds that
\[
 \begin{split}
\hat{D}_i(\gamma^{ij}\mathbb{S}^{\efindex}_{jk})&=\frac{1}{k^2}\gamma^{ij}\hat{D}_i\hat{D}_j\hat{D}_k \mathbb{S}^{\efindex}+\frac{1}{n}\hat{D}_k \mathbb{S}^{\efindex}
\\
&=-\hat{D}_k \mathbb{S}^{\efindex}+\frac{1}{k^2}\gamma^{\ell m}\hat{R}_{mk}\hat{D}_\ell\mathbb{S}^{\efindex}+\frac{1}{n}\hat{D}_k \mathbb{S}^{\efindex}
\\
&=\left(\frac{1}{n} - 1+ \frac{(n-1)K}{k^2}\right)\hat{D}_k \mathbb{S}^{\efindex}\,,
 \end{split}
\]
where we have used $\hat{R}_{ij}=(n-1)K \gamma_{ij}$ and
\[
\nabla_i\nabla_j\nabla_k f = \nabla_k \nabla_j \nabla_i f + R^\ell{}_{jki} \partial_\ell f\,.
\]

The double divergence of $\mathbb{S}^{\efindex}_{jk}$ is then
\[
\hat{D}_\ell \hat{D}_i (\gamma^{\ell k}\gamma^{i j}\mathbb{S}^{\efindex}_{jk})=-\left(k^2\left(\frac{1}{n}-1\right)+(n-1)K \right)\mathbb{S}^{\efindex}\,.
\]
For $k\ne 0$ this gives for the divergence and double divergence of $\delta \mathring{g}^S_{ij}$
\begin{equation}\label{13VII17.1}\begin{split}
\hat{D}_i(\gamma^{ij}\delta\mathring{g}_{jk}^S)&=
\sum_{\efindex}2r^2\left[H^{\efindex}_{L,S}\hat{D}_k\mathbb{S}^{\efindex}+H^{\efindex}_{T,S}(\hat{D}_i \gamma^{ij}\mathbb{S}^{\efindex}_{jk})\right]\\
&=\sum_{\efindex}2r^2\hat{D}_k\mathbb{S}^{\efindex} \left[H^{\efindex}_{L,S}+H^{\efindex}_{T,S}\left(\frac{1}{n}-1+\frac{(n-1)K}{k^2}\right)\right] \,,
\end{split}\end{equation}
and
\begin{equation}\label{13VII17.2}\begin{split}
\hat{D}_\ell\hat{D}_i(\gamma^{\ell k}\gamma^{ij}\delta\mathring{g}_{jk}^S)&=
\sum_{\efindex}-2r^2k^2\mathbb{S}^{\efindex}\left[H^{\efindex}_{L,S}+H^{\efindex}_{T,S}\left(\frac{1}{n}-1+\frac{(n-1)K}{k^2}\right)\right]\,.
\end{split}\end{equation}
Similarly, for $\mathbb{V}^{\efindex}_{\ell m}$, assuming $k_V\ne 0$,
\[\begin{split}
\hat{D}_j(\gamma^{\ell j}\mathbb{V}^{\efindex}_{\ell m})&=
-\frac{1}{2 k_V}\gamma^{\ell j}\left(\hat{D}_j\hat{D}_\ell \mathbb{V}^{\efindex}_m + \hat{D}_j\hat{D}_m \mathbb{V}^{\efindex}_\ell \right)\\
&=-\frac{1}{2k_V}\left(-k_V^2\mathbb{V}^{\efindex}_m+\hat{D}_m\gamma^{\ell j}\hat{D}_j \mathbb{V}^{\efindex}_\ell + \hat{R}^j_{\ell jm}\gamma^{\ell i}\mathbb{V}^{\efindex}_i\right)\\
&=\frac{1}{2} \mathbb{V}^{\efindex}_m(k_V+(n-1)K)\,,
\end{split}\]
where we have used $\hat{D}_j(\gamma^{jk}\mathbb{V}_k)=0$~\cite[(5.7b)]{KodamaIshibashiMaster}, and
\[
\hat{D}_i\hat{D}_j(\gamma^{m i}\gamma^{\ell j}\mathbb{V}^{\efindex}_{\ell m})=(k_V+(n-1)K)\hat{D}_i(\gamma^{mi}\mathbb{V}_m)=0\,.
\]
Therefore
\begin{equation}\label{13VII17.3}
\hat{D}_i(\gamma^{ij}\delta\mathring{g}_{jk}^V)=
\sum_{\efindex}r^2 H^{\efindex}_{T,V} \mathbb{V}^{\efindex}_m(k_V+(n-1)K)\,,
\end{equation}
and
\begin{equation}\label{13VII17.4}
\hat{D}_\ell\hat{D}_i(\gamma^{\ell k}\gamma^{ij}\delta\mathring{g}_{jk}^V)=0\,.
\end{equation}

%

\section{The decomposition of $h$}
 \label{appendicedecompo}

In this appendix we wish to justify  the decomposition \eq{7VI17.101} of $\deltagriem$.

Recall that we denote by $x^a$ the coordinates $t$ and $r$, and by $x^i$ the coordinates on the compact boundaryless Riemannian manifold $(N,\gamma)$.
The metric functions $h_{ab}$ form obviously a family of $t$- and $r$-dependent scalar functions on $N$, similarly for the $\gamma$-trace $\gamma^{ij}h_{ij}$ of $h_{\mu\nu}$.  Next, the fields $h_{ai}dx^i$ form a family of $t$- and $r$-dependent covectors on $N$, while $h_{ij}dx^i dx^j$ forms a similar family of tensor fields on $N$. Removing the $\gamma$-trace of $h_{ij}dx^i dx^j$, one obtains a family of trace-free tensor fields on $N$.

We have  the standard two $L^2$-orthogonal  decompositions   for vector or covector fields, and for trace-free symmetric two-tensor fields (see e.g.~\cite{BergerEbin1969}):
$$
C^{\infty}(N,TN)=\mathrm{Im}\; d\oplus\ker d^*,
$$
$$
C^{\infty}(N,\mathring S_2N)=\mathrm{Im}\;\mathring{\mathcal L}\oplus\ker \div,
$$
where $\mathring{\mathcal L}$ is the conformal Killing operator,
$$
 \mathring{\mathcal L}(W)_{ij}= \hat D_i W_j + \hat D_j W_i - \frac 2 n \hat D^k W_k \gamma_{ij}
  \,.
$$
This allows to write any covector field uniquely as
$$
W_i=W_i^S+W_i^V\;,\; W^S\in \mathrm{Im}\; d\;,\;W^V\in\ker d^*
 \,.
$$
This decomposition can be applied to the fields $h_{ai}dx^i$.

Next, any trace free symmetric tensor field $u_{ij}dx^idx^j$ can be written in unique way as
$$
u _{ij}=u ^S_{ij}+u _{ij}^V+u _{ij}^T \;,
\ u ^S\in\mathring{\mathcal L}(\mathrm{Im}\;d)
\;,
\
u ^V\in\mathring{\mathcal L}(\ker d^*)
\;,
\ u ^T \in\ker \div
 \,.
$$
We apply this last decomposition to the $\gamma$-trace-free part of $h_{ij}dx^idx^j$ which, together with what been said above, results in \eq{7VI17.101}.

Note that (minus) the  divergence of $u$ is also decomposed as
$$
\hat D^ju _{ij}=\hat D^ju ^S_{ij}+\hat D^ju ^V_{ij}
 \,.
$$
Next, if we use
(compare Appendix~\ref{sA17VIII1})
$$
\begin{array}{lll}
\div \mathring{\mathcal L}&=&\hat D^*\hat D_g-\Ric+\frac{n-2}2dd^*\\
&=&\displaystyle{\Delta_{\mbox{\tiny Hodge}}-2\Ric+\frac{n-2}2dd^*}\\
&=& \displaystyle{d^*d+\frac{n}2dd^*-2\Ric},
\end{array}
$$
and if the metric $\gamma$ is Einstein, we have
$$
\div u ^S\in \mathrm{Im} \,d\;,\;\; \div u ^V\in \ker d^*.
$$
More precisely, if $\Ric=K(n-1)$ and  $u ^S=\mathring{\mathcal L}dS$ then
$$
\div u ^S=d\left[\frac n2 (d^*d S)-2K(n-1)S\right],
$$
and if $u ^V=\mathring{\mathcal L}V$, with $d^*V=0$ then
$$
\div u ^V=d^*d V-2K(n-1)V.
$$


%


\section{The asymptotics of scalar modes}
 \label{s5VII17.6}

Whatever the value of $k$, for large $r$ the potentials
$V_{S,k }$ tend  to
$$V_{S,k }(\infty)=\frac{(n-2)(n-4)}4.$$
We find the characteristic exponents of the master equation by solving the equation
$$
s(1-s)+V_{S,k }(\infty)=0,
$$
giving
$$
 s_\pm=\frac{1\pm|n-3|}2
 \,.
$$
This has to be compared with the asymptotics
\bel{5VI17.4}
{\Phi}_{S,\modeindex  }=O(\rho^{(n-2)/2})
  \,,
\ee
as obtained by directly translating the asymptotic behaviour of elements of the $L^2$-kernel of the linearised Einstein operator into the asymptotic behaviour of the master fields.
For $n\geq 4$ the decay \eq{5VI17.4} corresponds to $s_+$.
For $n=3$ there is only one index, implying logarithmic terms in an asymptotic expansion.
For $n=2$ the decay \eq{5VI17.4} corresponds to $s_-$, which is zero, but we will improve this asymptotics below.

\subsection{Estimate of $\Phi_{S,\efindex}$ in any dimension $n$}
 \label{sec:falloff_alln}

From $|h|_{\mathring\griem}=O(\rho^{n+1})$ we obtain for the components in $(t,r,z^j)$ coordinates $(j=2,\dots, n+1)$
\begin{equation}
\begin{aligned}\label{20IX17.3}
h_{tt}&=O(r^{1-n})\,, & h_{tr}&=O(r^{-1-n})\,, &  h_{rr}&=O(r^{-3-n})\,,\\
 h_{tj}&=O(r^{1-n})\,, & h_{rj}&=O(r^{-1-n})\,, & h_{jk}&=O(r^{1-n})\,.
\end{aligned}
\end{equation}
By \eqref{06IX17.1}
\[
 f^S_{tt,\modeindex}=O(r^{1-n})\,,\quad f^S_{tr,\modeindex}=O(r^{-1-n}) \,,\quad f^S_{rr,\modeindex}=O(r^{-3-n})\,,
\]
\[
f^S_{t,\modeindex}=O(r^{-n})\,,\quad f^S_{r,\modeindex}=O(r^{-2-n})\,,\quad H^S_{L/T,\modeindex}=O(r^{-1-n})\,.
\]
Then we have for $X_{a,\modeindex}$, defined in \eqref{4VII17.25},  assuming $k\ne 0$,
\[
X_{t,\modeindex}=O(r^{1-n})\,,
\]\[
\quad X_{r,\modeindex}=\frac{r}{k}(f_r+\frac{r}{k}D_r H_T)=\frac r k \left(O(r^{-2-n})+\frac r k O(r^{-2-n})\right)=O(r^{-n})\,,
\]
and~\cite[Equation~(2.7)]{KodamaIshibashiMaster}
\[\begin{split}
F_{\modeindex}:=&\, \HLSI+\frac{\HTSI}{n}+r^{-1}D^ar X_{a,\modeindex}\\
=&\, O(r^{-1-n})+r^{-1}g^{rr}X_{r,\modeindex}\\
=&\, O(r^{-1-n})+O(r^{1-n})=O(r^{1-n})\,,\\
F_{ab,\modeindex}:=&\, \fabSI+D_a X_{b,\modeindex}+D_b X_{a,\modeindex},
\end{split}\]\[
F_{tt,\modeindex}=O(r^{1-n})\,,\quad F_{tr,\modeindex}=O(r^{-n})\,,\quad F_{rr,\modeindex}=O(r^{-n-1})\,.
\]
The rescaled quantities, defined as~\cite[(2.13)]{KodamaIshibashiMaster}
\[
\tilde{F}_{\modeindex}:=r^{n-2}F_{\modeindex}\,,\quad \tilde{F}_{ab,\modeindex}:=r^{n-2}F_{ab,\modeindex}\,,
\]
are
\[
\tilde{F}_{\modeindex}=O(r^{-1})\,,\qquad \tilde{F}_{tt,\modeindex}=O(r^{-1})\,,\quad \tilde{F}_{tr,\modeindex}=O(r^{-2})\,,\quad \tilde{F}_{rr,\modeindex}=O(r^{-3})\,.
\]
Then~\cite[(2.20)]{KodamaIshibashiMaster}
\[
X_{\modeindex}:=\tilde{F}^t_{t,\modeindex}-2\tilde{F}_{\modeindex}=O(r^{-1})\,,
\quad Y_{\modeindex}:=\tilde{F}^r_{r,\modeindex}-2\tilde{F}_{\modeindex}=O(r^{-1})\,,
\quad Z_{\modeindex}:=\tilde{F}^r_{t,\modeindex}=O(1)\,,
\]
and finally~\cite[(3.1)]{KodamaIshibashiMaster}
\[
\Phi_{S,\modeindex}:=
H_{\modeindex}^{-1} r^{1-n/2}
(n\tilde{Z}_{\modeindex} - r (X_{\modeindex} + Y_{\modeindex}))
=O(r^{1-n/2})=O(\rho^{n/2-1})\,,
\]
where
\[
H_{\modeindex}:=m+x n(n+1)/2=O(1)\,,
\]\[
m:=k(\modeindex)^2-n K\,,\quad x:=2\mu r^{1-n}
\]
and $\tilde Z$ is of the same order as $Z$.

\subsection{Estimate of $\Phi_S$ in  dimension $n=2$}
 \label{s5VII17.1}

We revisit the preceding equations when $n=2$.
From
\cite[(2.11),(2.13)]{KodamaIshibashiMaster}
we see that
$$
 \tilde{F}={F}\,,\;\;F^a{}_a=0
  \,.
$$
Let us write
$$
H_T=h_tr^{-3}+O(r^{-3-\epsilon})\;,\;\;  D_rH_T=-3h_tr^{-4}+O(r^{-4-\epsilon}),
$$
where $\epsilon>0$.
From the estimates in Section~\ref{sec:falloff_alln},
$$
X_r=-3\frac1{k^2}h_tr^{-2}+O(r^{-2-\epsilon})\;,\;\;F=-3\frac1{k^2}h_tr^{-1}+
O(r^{-1-\epsilon})
$$
$$
F_{tt}=O(r^{-1})\;,\;\;F_{tr}=O(r^{-2}),
$$
and
$$
F_{rr}=2D_r(r^2D_rH_T)+O(r^{-3-\epsilon})=12h_tr^{-3}+O(r^{-3-\epsilon})
 \,.
$$
From $F^a{}_a=0$ we see that $h_t=0$, so $F$ and then $X+Y$ decay faster. We deduce from the definitions of $X$ and $Y$ that there exists $\epsilon>0$
 such that
$$
 F,X,Y=O(r^{-1-\epsilon})
  \,.
$$
This information inserted into~\cite[(2.24d)]{KodamaIshibashiMaster}  gives
$$
\tilde Z=O(r^{- \epsilon})
\,.
$$
In conclusion
$$
\Phi_{S,\modeindex}=O(r^{- \epsilon}),
$$
tends to zero at infinity.

\section{The asymptotics of vector modes}
 \label{s5VII17.8}

In order to estimate the rate of decay of $\Phi_{V,\modeindex}$ we start with~\cite[(5.10)]{KodamaIshibashiMaster}
\begin{align*}
F_r&=f_r+\frac{r}{k_V}D_r H_T=O(r^{-2-n})+O(r^{-1-n}))=O(r^{-1-n})\,,\\
F_t&=f_t+\frac{r}{k_V}D_t H_T=O(r^{-n})\,.
\end{align*}
From~\cite[(5.12)]{KodamaIshibashiMaster} we have
\[
r^{n-1}F_a=\epsilon_a{}^b D_b\Omega\,,
\]
where $\epsilon_{ab}=O(1)$. Therefore
\[
D_t\Omega=O(1)\,,\quad D_r\Omega=O(r^{-3})\,,
\]
and $\Omega=O(1)$. By~\cite[(5.13)]{KodamaIshibashiMaster} we have
\[
\Phi_{V,\modeindex}=r^{-n/2}\Omega=O(r^{-n/2})=O(\rho^{n/2})\,.
\]
%

%

\section{The linearised Birkhoff theorem}
 \label{ss2IX17.1}

In the case $K=1$, $n=2$, a gauge transformation $h_{\mu\nu}\to h_{\mu\nu}+\mathcal{L}_Y \zgriem_{\mu\nu}$, with gauge vector $Y$, takes the form
\beal{18VIII17.1}
h_{tt}&\to&h_{tt}+ Y^r \partial_r \VtwoOrF +2 \VtwoOrF  \partial_t Y^t
\,,
\\
\label{27VIII17.1}
h_{tr}&\to&h_{tr}+ \VtwoOrF^{-1} \partial_t Y^r+\VtwoOrF  \partial_r Y^t
\,,
\\
\label{27VIII17.2}
h_{rr}&\to&h_{rr}+ Y^r \partial_r \VtwoOrF^{-1}+2\VtwoOrF^{-1} \partial_r Y^r
\,,
\\
\label{27VIII17.3}
h_{\theta\varphi}&\to&h_{\theta\varphi} + r^2\sin^2 \theta \partial_\theta Y^\varphi+r^2\partial_\varphi Y^\theta
\,,
\\
\label{27VIII17.4}
h_{\varphi\varphi}&\to&h_{\varphi\varphi}+ 2r^2 \sin^2 \theta (r^{-1}Y^r+ Y^\theta \cot \theta +\partial_\varphi Y^\varphi )
\,,
\\
\label{27VIII17.5}
h_{\theta\theta}&\to&h_{\theta\theta}+ 2r Y^r+2r^2 \partial_\theta Y^\theta
\,,
\\
\label{04IX17.1}
h_{t\theta}&\to&h_{t\theta}+ \VtwoOrF  \partial_\theta Y^t+r^2\partial_t Y^\theta
\,,
\\
\label{04IX17.2}
h_{t\varphi}&\to&h_{t\varphi}+ \VtwoOrF  \partial_\varphi Y^t + r^2 \sin^2 \theta \partial_t Y^\varphi
\,,
\\
\label{04IX17.3}
h_{r\theta}&\to&h_{r\theta}+ \VtwoOrF^{-1}\partial_\theta Y^r + r^2 \partial_r Y^\theta
\,,
\\
\label{04IX17.4}
h_{r \varphi}&\to&h_{r\varphi}+ \VtwoOrF^{-1} \partial_\varphi Y^r + r^2\sin^2\theta \partial_r Y^\varphi
\,.
\eea

We consider an $\ourellnormal=0$  linearised solution $h_{\mu\nu}$  of the Einstein equations, i.e.
\bel{2IX17.1}
 h_{ab}= h_{ab}(t,r)
 \,,
 \quad
 h_{ia} \equiv 0
 \,,
 \quad
 h_{ij} = \fpsi (t,r) \zgriem _{ij}
 \,.
\ee
We assume that $h\in L^2$ and is in the kernel of the operator $P_L$. We set
\begin{equation}\label{2IX17.6}
 Y^r =  r \fpsi   /2  = O(r^{-2})
 \,,
 \quad
 Y^i \equiv 0
\end{equation}
(recall that $O(r^{-2})$ is understood for large $r$),
which implies
$$
 h_{ij} = \mcL_Y \zgriem_{ij}
 \,.
$$
Recall that $r_0>0$ has been defined as the largest zero of $f$. Let $\epsilon>0$.
We define $Y^t_\epsilon$  by integrating \eq{27VIII17.1} in $r$ and cutting-off  near $r_0$:
\begin{equation}\label{2IX17.2}
  Y^t_\epsilon =   -  \chi_\epsilon(r)\int_{r }^\infty
   \VtwoOrF^{- 1}(h_{rt} -  \partial_t Y^r  \VtwoOrF^{-1}) dr
    = O(r^{-4})
    \,,
\end{equation}
where $\chi_\epsilon$ is a smooth function which is identically equal to one for $r\ge r_0+\epsilon$ and zero for $r\in[r_0,r_0+\half \epsilon]$.
(The vector field defined by the integral above without the cut-off  might be singular at $r=r_0$; we will see at the end of the argument that the cut-off was unnecessary, but this is not clear from the outset.
Note that if $\epsilon_1 > \epsilon_2 $ then $Y_{\epsilon_1}(r)=Y_{\epsilon_2}(r)$ for $r> r_0+\epsilon_1$.)

Then, for   $r>r_0+ \epsilon $,
$$
 h_{tr} = \mcL_{Y_\epsilon} \zgriem_{tr}
 \,.
$$
%
An analysis of the components of $\mcL_{Y_\epsilon} \zgriem$, using (\ref{18VIII17.1})-(\ref{04IX17.4}), proves that
$$
 |\mcL_{Y_\epsilon} \zgriem|^2_\zgriem=O(r^{-6}),
$$
so $\mcL_{Y_\epsilon} \zgriem\in L^2$.

Set
\begin{equation}\label{2IX17.3}
  \hat h_{\mu\nu} = h_{\mu\nu} - \mcL_{Y_\epsilon} \zgriem_{\mu\nu}
  \,,
\end{equation}
thus $\hat h_{\mu\nu}$ is a solution of the linearised Einstein equations with all components vanishing for large $r$ except possibly $h_{tt}$ and $h_{rr}$. Now, a variation of the metric arising from a variation $\delta \mu$ of the mass  in the coordinate system \eq{21III17.1} takes the form
\begin{equation}\label{2IX17.4}
    2\frac{\delta \mu }{r} ( - dt^2 +  \VtwoOrF^{-2} dr^2 )
  \,,
\end{equation}
which suggests that it might be convenient to define new functions $\tilde h_{rr}$ and $\tilde h_{tt}$ as
\begin{equation}\label{2IX17.5}
  \tilde h_{rr}:= r \VtwoOrF^2 \hat h_{rr}
  \,,
  \quad
  \tilde h_{tt}:= r (\hat h_{tt} + \VtwoOrF^2 \hat h_{rr})
  \,.
\end{equation}
Inserting \eqref{2IX17.5} into the linearised Einstein tensor $G'[h]_{\mu\nu}$ one finds, again for  $r>r_0+\epsilon $,
\begin{equation}\label{2IX17.8}
  G_{tr}'[h] = \frac{  \partial_t\tilde h_{rr} }{r
   \left(-2 \mu +r^3+r\right)}
    \,,
\end{equation}
thus $\tilde h_{rr}$  depends at most upon $r$. One can now eliminate the second radial derivative of $\tilde h_{tt}$ between the $G_{tt}$ and $G_{rr}$ equations, obtaining
\begin{equation}\label{5IX17.1}
   \partial_r\left(\frac{\tilde h_{tt}}{r \VtwoOrF  }\right)
    =0
   \,.
\end{equation}
Hence,
\begin{equation}\label{2IX17.9}
  \tilde h_{tt} = C(t) r\VtwoOrF
  \,.
\end{equation}
for   $r>r_0+\epsilon $, for some function $C$ depending only upon $t$.
Inserting all this into the equations $G_{ij}=0$ gives $\partial_r \tilde h_{rr}=0$, and thus $\tilde h_{rr}$ is a constant, say $2\delta \mu$.

In terms of $\hat{h}_{rr}$ and $\hat{h}_{tt}$ we now have, for  $r>r_0+\epsilon $,
\begin{equation}\label{4IV18.1}
	\hat{h}_{tt}=f C(t) -\frac{2\delta\mu}{r}\,,\quad \hat{h}_{rr}=\frac{2\delta\mu}{r f^2}\,.
\end{equation}
%

This is not in $L^2$ unless $C\equiv 0$. Further, our gauge-transformed field with $C(t)=0$ is a variation of the mass, and therefore, as shown in Section \ref{sA7IV17.1}, it is in $L^2$ if and only if the mass takes the critical value.

For all masses except the critical one the tensor field $h$ is therefore pure gauge for  $r>r_0+\epsilon $. Since $\epsilon$ is arbitrary, we conclude that for  $r>r_0$ we have
\begin{equation}\label{4VI18.21}
h_{\mu\nu} = \mcL_{Y_0} \zgriem_{\mu\nu}
\,,
\end{equation}
where $Y_0$ is defined as in \eqref{2IX17.2} with $\chi_0\equiv 1$. As $h$ has been assumed to be smooth, it is simple to show from \eqref{4VI18.21} that $Y_0$ is smooth across the rotation axis, so that $h$ is pure gauge everywhere, as desired.

If we denote by $\hSdiamond{}$ the $\ourellnormal=0$-part of $h $,
and  by $\YSdiamond{}$ the  gauge vector field $Y_0$ just constructed, we have, for $\mu\neq\mu_c$,
\begin{equation}\label{4IX17.41+}
  \hSdiamond{} = \mcL_{\YSdiamond{}} \zgriem
  \,,
  \qquad
   | \YSdiamond{}|_{\zgriem} = O(r^{-3})
  \,.
\end{equation}

\section{The $l=1$ modes for $K=1$, $n=2$}
 \label{s9IX17.1}

In this section we analyze the $\ourellnormal=1$ modes when $(\Nman,\gamma_{ij})$ is a two-dimensional round unit sphere.  We follow the treatment of the Lorentzian case in~\cite{Dotti2016}, which requires only trivial modifications when addressing the Riemannian setting. We present the argument here because of  the need  of establishing estimates for the gauge vector fields.

As explained in~\cite[end of Section~2.2]{IshibashiWaldIII} (compare~\cite[Section~II.A]{Dotti2016}),
in the case at hand a general metric perturbation can be decomposed as $h_{\mu\nu}=h^{(-)}_{\mu\nu}+h^{(+)}_{\mu\nu}$, where $h^{(-)}_{\mu\nu}$ is the \emph{odd} and $h^{(+)}_{\mu\nu}$ the \emph{even} part. The linearised Einstein equations decouple into separate equations for the odd and even parts.

\subsection{Odd perturbations}
 \label{ss11IX17.1}

Odd perturbations take the form~\cite[Equation~(35)]{Dotti2016}
\bel{08IX17.1}
h^{(-)}_{\mu\nu}=
\begin{pmatrix}
	0 &
	 r^2h_a^{(l,m)}\hat{\epsilon}_i{}^j\hat{\nabla}_j S_{(l,m)} \\
	r^2h_a^{(l,m)}\hat{\epsilon}_i{}^j\hat{\nabla}_j S_{(l,m)} &
	 2r^4 k^{(l,m)}\hat{\epsilon}_{(i}{}^m\hat{\nabla}_{j)}\hat{\nabla}_m S_{(l,m)}
\end{pmatrix}\,,
\ee
where $h_a^{(l,m)}$ and $k^{(l,m)}$ are functions of $t$ and $r$, $\hat{\nabla}$ is the covariant derivative on the sphere, $\hat{\epsilon}_{ij}=\sin \theta (\delta_{i}^{\theta}\delta_j^{\phi}-\delta_i^{\phi} \delta_j^{\theta})$
and $S_{(l,m)}$ are the scalar spherical harmonics.

Restricting to $l=1$ modes, the $ij$ components of $h^{(-)}_{\mu\nu}$ vanish and the $ai$ components take the form
\bel{08IX17.3}
h^{(-)}_{ai}=\sum_{m=-1}^1\sqrt{\frac{3}{4\pi}}h_a^{(l=1,m)}J_{(m)i}
\ee
where the $J_{(m)}$'s form a basis of Killing vector fields on $S^2$.

Gauge transformations defined by a gauge vector $Y$ of the form~\cite[(28)]{Dotti2016}
\[
Y^a=0\,,\quad Y^i=r^4 \hat{\epsilon}^{ij}\hat{\nabla}_j f_Y
 \\,
\]
for some function $f_Y$, preserve the odd character of the perturbation and those with
\[
f_Y=\sum_{m=-1}^1 f_Y^{(m)} S_{(l=1,m)}\,,
\]
stay within the $l=1$ modes. The effect of such a gauge transformation on the perturbation is given by~\cite[(73)]{Dotti2016}
\[
h^{(-)}_{ai}\to
\sum_{m=-1}^1\sqrt{\frac{3}{4\pi}}(h_a^{(l=1,m)}+r^2\partial_a f_Y^{(m)})r^{-2} J_{(m)i}\,,
\]
with all other components unaffected.

Defining $\hat{h}^{(-)}$ by $h^{(-)}_{\mu\nu}=\hat{h}^{(-)}_{\mu\nu}+\mathcal{L}_\Ydiamondm  {} \zgriem$ with a gauge vector $\Ydiamondm  {}$ defined as $(\Ydiamondm  {})^ a:=0$ and
\bel{08IX17.2}
 (\Ydiamondm  {})^ i=-r^{4}\hat{\epsilon}^{ij}\partial_j
 \big(
  \sum_{m=-1}^1 S_{(l=1,m)}\int h_r^{(l=1,m)} r^{-2}dr
   \big)
    =O(r^{-4})\,,
\ee
the components $\hat{h}^{(-)}_{ri}$ vanish, leaving only $\hat{h}^{(-)}_{ti}$. The norm of the gauge part is found to be
\[
|\mathcal{L}_\Ydiamondm  {}\zgriem|^2_{\zgriem}=O(r^{-6})\,,
\]
and, as $f_Y^{(m)}$ is regular at $r_0$ and therefore $f^{-1}\partial_t f_Y^{(m)}$ is as well, $\mathcal{L}_\Ydiamondm  {}\zgriem\in L^2$.

As the $h^{(-)}_{ti}$-components are of the form \eqref{08IX17.3}, we can set one of them to zero by a rotation, leaving $\hat{h}^{(-)}_{t\varphi}=\sqrt{\frac{3}{4\pi}}\hat{h}^{(l=1,0)}_t  \sin^2\theta$,
for some functions $\hat{h}^{(l=1,0)}_t (t,r)$,
as the only non-zero component of the perturbation.

The linearised Einstein equations give
\beal{08IX17.4}
r^2 \partial_r^2 \hat{h}^{(l=1,0)}_t -2 \hat{h}^{(l=1,0)}_t=0
\,,\\ \label{08IX17.5}
2 \partial_t \hat{h}^{(l=1,0)}_t-r \partial_r\partial_t \hat{h}^{(l=1,0)}_t=0
\,.
\eea
%
Integrating \eqref{08IX17.5} twice, we obtain
\[
\hat{h}^{(l=1,0)}_t=C_1 r^2 t + g(r)\,,
\]
where $C_1$ has to vanish because of the periodicity of the $t$ coordinate. Inserting this into \eqref{08IX17.4} leads to
\[
\hat{h}^{(l=1,0)}_t=C_2 r^2+C_3 r^{-1}
 \,.
\]
As the tensors $dt dx^i$ are not smooth at the axis of rotation $r=r_0$ we require $C_2 r_0^2+C_3 r_0^{-1}=0$, i.e.
\[
\hat{h}^{(l=1,0)}_t=C_2\frac{r^3-r_0^3}{r}\,.
\]
Perturbations of this form are exactly the variations of angular momentum within the Riemannian Kerr anti-de Sitter family which are analyzed in Section \ref{ss3IX17.1}.
As these variations are not in $L^2$ we have $\hat{h}=0$,
 and if we denote by $
  h^{\diamond 1 -} $ the odd part of the $l=1$ modes, we obtain
\begin{equation}\label{10IX17.12}
  h^{\diamond 1 -} = \mcL_{\Ydiamondm{}} \zgriem
  \,,
  \quad
  |\Ydiamondm{}|_\zgriem = O(r^{-3})
   \,.
\end{equation}

\subsection{Even perturbations}
 \label{ss11IX17.2}

Even $\ourellnormal=1$ solutions of the linearised Einstein equations can be parameterised as~\cite[Equation~(36)]{Dotti2016}
\begin{equation} \label{pert+}
 (h_{\alpha \beta}^{(+)} ) = \left( \begin{array}{cc}
 h_{ab}^{(\ell=1)}
  &
   \hat D_i q_a^{(\ell=1)}
\\
 \hat D_i q_b^{(\ell=1)}
  &
   \frac{r^2}{2} J^{(\ell=1)} \gamma_{ij}
 \end{array} \right).
\end{equation}
Under gauge transformations
with gauge-vector $Y$ of the form
\begin{equation} \label{vec+}
(Y_\alpha ) = (Y_a, r^2\; \hat D_i X)
 \,,
\end{equation}
$(h_{\alpha \beta}^{(+)} )$ transforms to $(\hat{h}_{\alpha \beta}^{(+)} )$ given by~\cite[Equation~(122)]{Dotti2016}
\begin{equation} \label{pert+gauge+1}
\left( \begin{array}{cc}
h_{ab}^{(\ell=1)} + \tilde D_a Y_b^{(\ell=1)}  + \tilde D_b Y_a^{(\ell=1)}  & \hat D_i (q_a^{(\ell=1)} + r^2 \tilde D_a X^{(\ell=1)} + Y_a^{(\ell=1)})
\\
 \hat D_i (q_b^{(\ell=1)} + r^2 \tilde D_b X^{(\ell=1)} + Y_b^{(\ell=1)})
  &\;\;\;\;
   \frac{r^2}{2} \gamma_{ij}  (J^{(\ell=1)}
  - 4 X^{(\ell=1)}  + \frac{4}{r} Y^a_{(\ell=1)}  \tilde D_a r )
 \end{array}
  \right)
\,.
\end{equation}

According to~\cite[Section~IV.A.2]{Dotti2016}, it is convenient to define $(X,Y_a)$ by solving the following system of equations:
\begin{eqnarray}
    & r^2D_tX^{(\ell=1)}+Y_t=-q_t^{(\ell=1)}=O(r^{-1})
     \,,&
 \label{11IX17.11}
\\
   &
 r^2D_rX^{(\ell=1)}+Y_r=-q_r^{(\ell=1)}=O(r^{-3})
  \,,
   &
 \label{11IX17.12}
\\
    &
 D^b Y_b=-\frac12
  \zgriem^{ab}h_{ab}^{(\ell=1)}=O(r^{-3})
 \,.
 &
 \label{11IX17.13}
\end{eqnarray}
%
%
 With this choice, $\hat{h}^{(+)}$ satisfies
\bel{19IX17.07}
	\hat{h}^{(+)}_{ai}=0\,,\qquad
	\zgriem^{ab}\hat{h}^{(+)}_{ab}=0\,.
\ee
Note that \eqref{11IX17.11}-\eqref{11IX17.13} imply
\begin{equation}\label{11IX17.14}
 D^b(r^2D_bX)=\frac12\zgriem^{ab}h_{ab}-D^bq_b
 \,.
\end{equation}
%
The homogeneous version of the equation \eqref{11IX17.14} for $X$ has no non-trivial solutions tending to zero at infinity by the maximum principle. The operator at the left-hand side of \eqref{11IX17.14} has indicial exponents in $\{0,-3\}$, and therefore \eqref{11IX17.14} has a unique solution $X = O(r^{-3})$ which is a linear combination of $\ourellnormal=1$ spherical harmonics.

%
The conditions \eqref{19IX17.07} do not fix the gauge uniquely: an additional gauge transformation satisfying
\bel{19IX17.01}
r^2 D_a X+Y_a=0\,,\quad D^a Y_a=0\,,
\ee
preserves the form of $\hat{h}^{(+)}$.

By a rotation we set $J=J^{(1)}S_{(l=1,m=1)}$. We define new variables $C_a$ as
\begin{equation}\label{tt1s}
{\hat{h}}_{ab}^{(\ell=1,T)} = \frac{1}{\VtwoOrF} \left[  C_a^{(\ell=1)} D_b r + C_b^{(\ell=1)} D_a r - \tilde g_{ab} C_d^{(\ell=1)} D^d r \right], \;\;\;
C_a^{(\ell=1)} = {\hat{h}}_{ab}^{(\ell=1,T)} D^b r,
\end{equation}
and decompose them into modes
\bel{18IX17.1}
C_a=\sum_{m=1}^3 C^{(m)}_a S_{(l=1,m)}\,.
\ee
Eliminating $C_r$ between the $m=2,3$ components of $G_{tr}'[h]$ and $G_{t\phi}'[h]$ we obtain
\bel{18IX17.2}
\partial_r C_t^{(2,3)}+\frac{C_t^{(2,3)}}{r-2 \mu+r^3}=0\,,
\ee
and therefore, by integration from the axis of rotation at $r_0$, the $C_t^{(2,3)}$ depend only on $r$. The $m=2,3$ components of $G_{t\phi}'[h]$ show $\partial_t C_r^{(2,3)}=-\partial_r C_t^{(2,3)}$ which implies that $C_t^{(2,3)}$ vanishes and $C_r^{(2,3)}$ can only depend on $r$. Eliminating second derivatives between $G_{tt}'[h]$ and $G_{rr}'[h]$ shows that $C_r^{(2,3)}$ vanishes as well.

To handle the $m=1$ equations we define new variables $Z_r^{(1)}(t,r)$ and $Z_t^{(1)}(t,r)$ by
\bel{18IX17.3}
C_r^{(1)}=6 \mu \partial_r Z_r^{(1)}+\frac{r}{2}\partial_r J^{(1)}\,,\quad
C_t^{(1)}=Z_t^{(1)}-\frac{3\mu-r}{2}\partial_t J^{(1)}-r \VtwoOrF 6\mu \partial_t\partial_r Z_r^{(1)}\,.
\ee
Note that this defines $Z_r^{(1)}$ only up to an irrelevant constant.

The $m=1$ component of $G_{tr}'[h]$ directly gives $Z_t^{(1)}=0$. Eliminating third order derivatives from the remaining equations we obtain
\bel{18IX17.4}
\partial_r J^{(1)}+4 r\VtwoOrF \partial^2_r Z_r^{(1)}+\left(12 r^2+8\right) \partial_r Z_r^{(1)}=0\,.
\ee
Differentiating the $m=1$ equations by $r$ and using \eqref{18IX17.4} to express derivatives of $J$ by $Z_r^{(1)}$ gives two fifth order and one fourth order equation for $Z_r^{(1)}$. Eliminating higher derivatives we finally obtain a third order equation for $Z_r^{(1)}$
\bel{19IX17.05}\begin{split}
r^2\VtwoOrF \partial_r D^a(r^2 D_a Z_r^{(1)}) +r^4\partial_r(r^{-2}\VtwoOrF ) D^a(r^2 D_a Z_r^{(1)})=0\,.
\end{split}\ee
This implies
\bel{19IX17.06}
D^a(r^2 D_a Z_r^{(1)})=\frac{Cr^2}{\VtwoOrF}\,,
\ee
with a constant $C$ which has to vanish for $Z_r^{(1)}$ to be regular at $r_0$.

We now consider the remaining gauge freedom. We see from \eqref{19IX17.01} that for any $X$ satisfying
\bel{19IX17.03}
D^a(r^2 D_a X)=0
\ee
there exists an associated $Y_a$ giving a gauge transformation which preserves \eqref{19IX17.07}.

Inserting the definition of our new variables into \eqref{pert+gauge+1} we find that under a gauge transformation satisfying \eqref{19IX17.01} the fields $Z_r^{(1)}$ and $J^{(1)}$ transform as
\beal{19IX17.02}
 \partial_r Z_r^{(1)} &\mapsto& \partial_r( Z_r^{(1)}+X^{(1)})
 \,,
\\
 J^{(1)} &\mapsto& J^{(1)}-4X^{(1)}+\frac{4\VtwoOrF}{r}Y^{(1)}_r\,,
\eea
where $X$ and $Y_a$ have been split into $l=1$ modes as for $C_a$ in \eqref{18IX17.1}.

As $Z_r^{(1)}$ satisfies \eqref{19IX17.03} we can render it independent of $r$ by a gauge transformation with $\partial_r X^{(1)}=-\partial_r Z_r^{(1)}$, which preserves \eqref{19IX17.07}.

With $\partial_r Z_r^{(1)}\equiv 0$ we see from \eqref{18IX17.4} that $J^{(1)}$ can only depend on $t$. From the remaining equations $\partial_t^2 J^{(1)}=0$, i.e. $J^{(1)}$ is constant.

We can exploit the remaining freedom in $X^{(1)}$ to set
\bel{20IX17.1}
X^{(1)}=\frac{J^{(1)}}{4}\,,\qquad Y_a=0\,,
\ee
obtaining $J^{(1)}\equiv 0$.
	This gives $Z_r^{(1)}=\text{const}$ and therefore $C_r\equiv C_t\equiv 0$.

We arrive at $h^{(+)}=\mathcal{L}_{\Ydiamond{}^{+}}\zgriem$ where $\Ydiamond{}^{+}$ is the combined gauge vector consisting of the part defined by \eqref{11IX17.11}--\eqref{11IX17.13}, that given by \eqref{19IX17.02} and that given by \eqref{20IX17.1}. From the asymptotics \eqref{20IX17.3} of $h$ and  from \eqref{18VIII17.1}--\eqref{04IX17.4}
with the right-hand sides being zero now,
we conclude that
\bel{20IX17.4}
 h^{\diamond 1 +} \equiv h^{(+)}= \mcL_{\Ydiamond{}^+} \zgriem
  \,,
  \quad
 |\Ydiamond{}^{+}|_\zgriem=O(r^{-1})
 \,.
\ee

An alternative, and somewhat simpler, proof can be given using~\cite{JezierskiWaluk2}. Since the last reference is only available in Polish so far~\cite{Waluk}, we felt it more appropriate to provide the argument above.

%
\section{A linearised Birkhoff-type theorem with $\kappahere=0$}
 \label{s7VII17.101}

In the case $K=0$  the gauge transformations take the form
\beal{06IX17.10}
h_{tt}&\to& h_{tt}+ Y^r \partial_r \VtwoOrF+2 \VtwoOrF  \partial_t Y^t
\,,
\\
\label{06IX17.11}
h_{tr}&\to& h_{tr}+ \VtwoOrF^{-1} \partial_t Y^r+\VtwoOrF  \partial_r Y^t
\,,
\\
\label{06IX17.12}
h_{rr}&\to&h_{rr}+ Y^r \partial_r \VtwoOrF^{-1}+2\VtwoOrF^{-1} \partial_r Y^r
\,,
\\
\label{06IX17.13}
h_{ij}&\to&h_{ij}+ \delta_{ij} 2r Y^r +r^2 (\partial_i Y^j+\partial_j Y^i)
\,,
\\
\label{06IX17.16}
h_{ti}&\to&h_{ti}+ \VtwoOrF  \partial_i Y^t+r^2\partial_t Y^i
\,,
\\
\label{06IX17.18}
h_{ri}&\to&h_{ri}+ \VtwoOrF^{-1}\partial_i Y^r + r^2 \partial_r Y^i
\,.
\eea

We consider a $k=0=k_V=k_T$ linearised solution of the Einstein equations, i.e. one that only depends on $t$ and $r$:
\bel{07IX17.1}
h_{\mu\nu}=h_{\mu\nu}(t,r)\,.
\ee
As there are (constant) harmonic vectors on the background with $K=0$ we have to consider both scalar and vector perturbations (the tensor ones are always controlled by the master functions).  Scalar perturbations can be treated analogously to the $K=1$ case in Appendix \ref{ss2IX17.1}: They take the form \eqref{2IX17.1} and by choosing a gauge vector $Y$ as in \eqref{2IX17.6} and \eqref{2IX17.2}, we can set, for $r>r_0+\epsilon$, $\hat{h}_{ij}\equiv 0$ and $\hat{h}_{tr}\equiv 0$ where $\hat{h}_{\mu\nu}=h_{\mu\nu}-\mathcal{L}_Y\zgriem_{\mu\nu}$ and $\epsilon>0$ is an arbitrary cutoff distance.

We define $\tilde{h}_{rr}$, $\tilde{h}_{tt}$ as
\begin{equation}\label{12IX17.1}
\tilde h_{rr}:= r^{n-1} \VtwoOrF^2 \hat h_{rr}
\,,
\quad
\tilde h_{tt}:= r^{n-1} (\hat h_{tt} + \VtwoOrF^2 \hat h_{rr})
\,,
\end{equation}
and insert into the linearised Einstein tensor $G'_{\mu\nu}[\hat h]$ using the equations in~\cite[Appendix B]{KodamaIshibashiSeto}. For $r>r_0+\epsilon$ the $G'_{tt}[h]$, $G'_{tr}[h]$, $G'_{rr}[h]$ components lead to
\beal{07IX17.4}
\tilde{h}_{rr}=\tilde{h}_{rr}(r)\,,\\
\tilde{h}_{tt}=C(t) r\VtwoOrF
\eea
where $C$ has to vanish for $h$ to be in $L^2$. Inserting this into $G'_{ij}[h]$ gives $\tilde{h}_{rr}=2\delta \mu$, for some constant $\delta \mu$.

The only remaining perturbations, up to gauge, are variations of the mass. As these are never in $L^2$ for $K=0$, the scalar part of the perturbation has to be pure gauge.
	
We now consider the vector part, which takes the form
\begin{equation}\label{09IV18.1}
h_{ab}=0\,,\quad h_{ai}=h_{ai}(r)\,,\quad h_{ij}=0\,.
\end{equation}

We choose a gauge vector, which we denote by  $\Ydiamond{}$, such that $\Ydiamond{}^a\equiv 0$ and $\Ydiamond{}^i=\Ydiamond{}^i(t,r)$.
We define $(\Ydiamond{})^i$ by integrating \eqref{06IX17.18}
to obtain $h_{ri}=\mathcal{L}_Y\zgriem_{ri}$:
\bel{07IX17.3}
(\Ydiamond{})^i=\int r^{-2} \left(h_{ri}
-\VtwoOrF^{-1}\partial_i (\Ydiamond{})^r \right) dr
=\int r^{-2} h_{ri} dr = O(r^{-n-2})\,.
\ee
Equations \eqref{06IX17.10}--\eqref{06IX17.18} show that
$$
|\mcL_{\Ydiamond{}} \zgriem|^2_\zgriem=O(r^{-2n-2}),
$$
and, as it is regular at $r_0$, $\mcL_{\Ydiamond{}} \zgriem\in L^2$.

The nontrivial linearised Einstein equations turn out to be
\beal{07IX17.5}
r^2 \partial_r^2 h_{ti}+r (n-2) \partial_r h_{ti}-2 (n-1) h_{ti}=0
\,,\\ \label{07IX17.7}
2 \partial_t \hat{h}_{ti}-r \partial_r\partial_t \hat{h}_{ti}=0
\,.
\eea

Integrating \eqref{07IX17.7} twice gives
\[
\hat{h}_{ti}=C_1 r^2 t + g(r)\,,
\]
where $C_1$ has to vanish because of the periodicity of the $t$ coordinate. Inserting this into \eqref{07IX17.5} leads to
\[
\hat{h}_{ti}=C_2 r^2+C_3 r^{1-n}
\]
where $C_2$ has to vanish to ensure $\hat{h}_{ti}$ is in $L^2$, while $C_3$ has to vanish because the tensors
$dt dx^i$ are not smooth at the axis of rotation $r=r_0$.

Thus $\hat{h}_{\mu\nu}\equiv 0$ and, if we denote by $\hdiamond{}$ the part of $h $ from which all higher modes have been removed (in the notation of \eq{06IX17.1}, and setting $\mathbb{S}^0(x^i)\equiv 1$ for simplicity,
\begin{equation}\label{7IX17.401}
\hdiamond{\mu\nu} dx^\mu dx^\nu= f_{ab,0}^S dx^a dx^b +
r f_{ai,0}^V dx^a dx^i +
2r^2( H^S_{L,0} \gamma_{ij} + H^T_{T,0} \mathbbm{T}^0_{ij})dx^i dx^j
\,,)
\end{equation}
we obtain
\begin{equation}\label{07IX17.11}
\hdiamond{}= \mcL_{\Ydiamond{}} \zgriem
\,,
\qquad
| \Ydiamond{}  | _\zgriem=O(r^{-n-1})
\,.
\end{equation}
%

%
\section{A linearised Birkhoff-type theorem with $\kappahere=-1$}
 \label{s03X17.1}

In the case $\kappahere=-1$  the gauge transformations take the form
\beal{03X17.01}
h_{tt}&\to& h_{tt}+ Y^r \partial_r \VtwoOrF+2 \VtwoOrF  \partial_t Y^t
\,,
\\
\label{03X17.02}
h_{tr}&\to& h_{tr}+ \VtwoOrF^{-1} \partial_t Y^r+\VtwoOrF  \partial_r Y^t
\,,
\\
\label{03X17.03}
h_{rr}&\to&h_{rr}+ Y^r \partial_r \VtwoOrF^{-1}+2\VtwoOrF^{-1} \partial_r Y^r
\,,
\\
\label{03X17.04}
h_{ij}&\to&h_{ij}
+\delta_{ij}\left(Y^r\frac{2r}{(x^{n+1})^2}-Y^{n+1}\frac{2r^2}{(x^{n+1})^3}\right)
\\ \nn
&&\phantom{h_{ij}}+\frac{r^2}{x^{n+1}} (\partial_i Y^j+\partial_j Y^i)
\,,
\\
\label{03X17.05}
h_{ti}&\to&h_{ti}+ \VtwoOrF  \partial_i Y^t+\frac{4r^2}{B^2}\partial_t Y^i
\,,
\\
\label{03X17.06}
h_{ri}&\to&h_{ri}+ \VtwoOrF^{-1}\partial_i Y^r + \frac{4r^2}{B^2} \partial_r Y^i
\,,
\eea
where we use the form $h_\kappahere=\sum_{i=2}^{n+1} (dx^i)^2 / (x^{n+1})^2$ of the hyperbolic metric.

 We consider the scalar part of a $l=0$ (i.e.\ $k=0$) solution of the linearised Einstein equations in dimension $n=2$. This will cover  the full $l=0$ case for $\kappahere=-1$, as the tensor part is controlled by the master functions and there are no  vector modes with $l=0$, as described in Section \ref{s5VII17.2}.

 Such a perturbation takes the same form as for the $\kappahere=1$ case, i.e.
\bel{03X17.07}
 h_{ab}= h_{ab}(t,r)
 \,,
 \quad
 h_{ia} \equiv 0
 \,,
 \quad
 h_{ij} = \fpsi (t,r) \zgriem _{ij}
 \,.
\ee
We define a gauge vector $\Ydiamond{\epsilon}$, according to \eqref{2IX17.6} and \eqref{2IX17.2}. This gives, for $r>r_0+\epsilon$ where $\epsilon>0$ is an arbitrary cutoff distance, $h_{tr}=\mathcal{L}_{\Ydiamond{\epsilon}} \zgriem_{tr}$ and $h_{ij}=\mathcal{L}_{\Ydiamond{\epsilon}}\zgriem_{ij}$.

As in \eqref{2IX17.3} we define $\hat{h}_{\mu\nu}$ as $h_{\mu\nu} - \mcL_{\Ydiamond{}} \zgriem$, then $\tilde{h}_{rr}$, $\tilde{h}_{tt}$ are defined as for $K=1$,
\begin{equation}\label{03X17.10}
\tilde h_{rr}:= r \VtwoOrF^2 \hat h_{rr}
\,,
\quad
\tilde h_{tt}:= r (\hat h_{tt} + \VtwoOrF^2 \hat h_{rr})
\,,
\end{equation}
and all this is inserted into the linearised Einstein tensor $G'_{\mu\nu}[\hat h]$. We find that, for $r>r_0+\epsilon$
\beal{03X17.11}
\tilde{h}_{rr}=2\delta\mu\,,
\quad
\tilde{h}_{tt}=C(t) r\VtwoOrF
 \,,
\eea
for some constant $\delta\mu$. As in the $K=1$ case, $C\equiv0$ as $h$ is assumed to be in $L^2$ and, up to gauge, the only perturbations remaining are variations of the mass. For the case $K=-1$ these are never in $L^2$.

We conclude that $\Ydiamond{0}$ is smooth everywhere and, if we denote again  by $\hdiamond{}$ the part of $h$ from which all higher modes have been removed,
we obtain
\begin{equation}\label{09X17.03}
\hdiamond{}= \mcL_{\Ydiamond{}} \zgriem
\,,
\qquad
| \Ydiamond{}  | _\zgriem=O(r^{-3})
\,.
\end{equation}
%

%


\section{The Riemannian  Kerr anti-de Sitter metrics}
 \label{ss3IX17.1}

The Riemannian  Schwarzschild-anti de Sitter metrics belong to the family of the Riemannian  Kerr anti-de Sitter metrics, parameterised with $m$ and an ``angular momentum parameter'' $a$. The variations of those metrics with respect to the parameter $a$ provide non-trivial solutions of the linearised Einstein equations at the Schwarzschild-anti de Sitter metric, which we need to analyse. For this is it is convenient to start with a discussion of the family of the Riemannian  Kerr anti-de Sitter metrics with  small parameter $a$. Our presentation follows~\cite{ChHoerzinger}, where Kerr-Newman-de Sitter metrics were considered.

In Boyer-Lindquist coordinates, after the replacements $a \rightarrow i a $ and $ t \rightarrow it$  the Kerr anti-de Sitter  metric becomes
\begin{eqnarray}
\griem&=& \frac{\Sigma}{\Delta_r} dr^2 +  \frac{\Sigma}{\Delta_\theta} d\theta^2 + \frac{\sin^2 (\theta)}{\Xi^2 \Sigma} \Delta_\theta ( a d t + (r^2 -a^2) d \varphi)^2 \nonumber \\
 &\phantom{=}&+ \frac{1}{\Xi^2 \Sigma} \Delta_r (dt - a \sin^2(\theta) d \varphi)^2\,,
\label{15III15.9}
\end{eqnarray}
where, after setting $\lambda = \Lambda/3$, we have
$
\Sigma=r^2-a^2 \cos^2(\theta)$ and
$$ \Delta_r=(r^2-a^2)\left( 1- \lambda r^2 \right) -2 \mhere  r  \,,
$$
$
\Delta_\theta = 1 - \lambda a^2 \cos^2(\theta)$, and $\Xi=1 - \lambda a^2$.

Keeping in mind that we are interested in the metric for small $a$, we consider $\mhere>0$ and we assume that the largest zero of $\Delta_r$, which we denote by $r_0$, is positive.
For $r\in[r_0,\infty)$ we introduce a new coordinate  $\rho $  defined as
\begin{eqnarray}
\rho = \int_{r_0}^r \frac{1}{\sqrt{\Delta_r}}  dr
  = \frac{2}{  \sqrt{ \kappa  }}\sqrt{ (r-r_0)}\mathbbm{1}_{1 }(r-r_0)\,,
\end{eqnarray}
where
\bel{10V15.11}
 \kappa:=  \left|\Delta_{r} ' |\right|_{r={r_0}} \neq 0
 \,,
\ee
and with a function  $\mathbbm{1}_{1}$ which is smooth near the origin and satisfies $\mathbbm{1}_{1}(0)=1$.
Inverting, it follows that
\begin{eqnarray}
r=r_0 + \frac{\kappa}{4} \rho ^2 \mathbbm{1}_2(\rho ^2)
 \,,
 \qquad
 \Delta_r =  \frac{\kappa^2}{4}  \rho ^2 \mathbbm{1}_3(\rho ^2)
 \,,
\end{eqnarray}
with functions $\mathbbm{1}_2$, $\mathbbm{1}_3$ which are smooth near the origin, with $\mathbbm{1}_2(0)=1=\mathbbm{1}_3(0)$.

Smoothness at $r=r_0$ requires that  $t$ defines a $2\pi \omega$-periodic coordinate, with
\bel{7V15.4}
    \omega
  := \frac{ 2 \Xi \left(r_0^2-a^2\right) }{ \Delta_{r} ' (r_0)  }
    \,.
\ee
In order to guarantee regularity near the intersection of the axis $\{\sin \theta=0\}$ with the axis $\{\Delta_r=0\}$, near $\theta=0$ we use a coordinate system $(\rho ,\tau,\theta,\phi )$,
with $t=\omega  \tau$  and $\phi $ defined through the formula
\bel{5V15.1}
d \varphi
:= \alpha  d\phi  +  \frac{a  }{ a^2-r_0^2 \color{black}} d t
  \equiv \alpha  d\phi  +  \frac{a  \omega }{ a^2-r_0^2 \color{black}} d \tau
 \,,
\ee
for some constants $\alpha ,\omega  \in \R^*$ which will be determined shortly by requiring $2\pi$-periodicity of $\tau$ and $\phi $.
In this coordinate system the metric takes the form
\bean
 g
  &  = &
   {\Sigma} \bigg\{
    d \rho ^2
    + \frac{1}{\Xi^2 \Sigma ^2}
        \bigg[ \frac{   \kappa^2 \omega ^2 \Sigma ^2}{4\left( r_0^2-a^2 \color{black} \right)^2} \mathbbm{1}_4(\rho  ^2,\sin^2(\theta)) \rho  ^2 d \tau^2
\\
 \nonumber
  &&
   \phantom{ xx}
        + \alpha  ^2 \big(\Delta_\theta \left(a^2-r ^2\right)^2
         +a^2 \Delta_r \sin ^2(\theta ) \big)
         \sin^2(\theta) d \phi ^2
\\
  &&
   \phantom{ xx}
 + F(\rho^2,\sin^2(\theta))
   \rho  ^2 \sin   ^2(\theta )
   d \tau d \phi
    \bigg]
 +  \frac{1}{\Delta_\theta} d  \theta^2
 \bigg\}
 \,,
\eeal{6V15.2}
for some smooth functions $\mathbbm{1}_4$ and $F$, with $\mathbbm{1}_4(0,y)=1$.
As is well known, when $(\rho ,\tau)$ are viewed as polar coordinates around $\rho =0$, the one form $\rho ^2 d \tau$ and the quadratic form $d\rho ^2 + \rho ^2 d \tau^2$ are smooth.
Similarly when $(\theta,\phi )$ are  polar coordinates around $ \theta=0$, the one form
$\sin^2(\theta)d \phi $ and the quadratic form $d\theta^2 + \sin(\theta)^ 2d \phi ^2$ are smooth. It is then easily inferred that
the requirements of $2\pi$-periodicity of $  \tau$ and $\phi $, together with
\bel{7V15.2}
 \frac{   \kappa^2 \omega^2 }{4\Xi^2\left( r_0^2-a^2 \color{black}\right)^2} =1
\,,
\quad
 \frac{ \alpha  ^2 \Delta_\theta^2 \left(a^2-r ^2\right)^2 }{\Xi^2 (r^2-a^2 \cos^2(\theta))^2}\bigg|_{  \theta =0} \equiv \alpha ^2= 1
 \,,
\ee
imply smoothness both of  the sum of the diagonal terms of the metric $g$  and of the off-diagonal term $g_{\tau\phi } d \tau d \phi $  on
$$
 \Omega:= \{(r,\tau,\theta,\phi )\in[r_0,\infty)\times S^1\times [0,\pi)\times S^1\}
 \,.
$$
Here $[r_0,\infty)\times S^1$ is understood as $\R^2$ with center of rotation at $r_0$, similarly $[0,\pi)\times S^1$ is understood as a disc $D^2$  of radius $\pi$.

The above calculations remain valid without changes near $\theta=\pi$.   When $\theta\in(0,\pi]$ we will denote by $\hat \tau$ and $\hat \phi $   the relevant angular coordinates, and $\hat \omega $, $\hat \alpha $ the corresponding coefficients. Thus,  for $\theta\in(0,\pi]$:
\bel{15V15.51}
 t=\hat \omega  \hat \tau\,,
 \quad
  d \varphi= \hat \alpha  d\hat \phi  +  \frac{a \hat \omega }{ a^2-r_0^2 \color{black}} d \hat \tau
 \,,
\ee
with
\bel{15V15.52}
\hat \omega  = \pm \omega
\,,
\quad
 \hat \alpha    = \pm 1
 \,.
\ee
We obtain likewise a smooth metric on the set
$$
 \widehat \Omega:= \{(r,\hat \tau,\theta,\widehat \phi )\in[r_0,\infty)\times S^1\times (0,\pi]\times S^1\} \approx \R^2 \times D^2
 \,.
$$

Disregarding issues of orientation, without loss of generality we can choose the plus signs above.
The manifold $M$,
obtained by patching together $\Omega$ with $\widehat \Omega$, using the obvious identifications resulting from the formulae
\bel{15V15.61}
\omega d\tau=\hat \omega  d\hat \tau\,,
 \quad
     \alpha  d  \phi  +  \frac{a   \omega }{ a^2-r_0^2 \color{black}} d  \tau
      = \hat \alpha  d\hat \phi  +  \frac{a \hat \omega }{ a^2-r_0^2 \color{black}} d \hat \tau
 \,,
\ee
is diffeomorphic to $\R^2 \times S^2$.

Note that while $d\varphi$ is a well defined one-form on $M$, the function $\varphi$ is a well defined coordinate-modulo-$2\pi$ on $M$ if and only if  $\frac{a  \omega }{ a^2-r_0^2 \color{black}}\in \Z^*$. We emphasise that it is \emph{not} necessary to impose this last restriction to obtain a well defined smooth Riemannian metric on $M$, and we will not impose it.

Differentiating \eq{15III15.9} with respect to $a$ in the $(\rho,\tau,\theta,\phi)$ coordinates we obtain
\begin{eqnarray}
  \label{kds.metric-a}
    \frac{dg}{da}\big|_{a=0} &=& -
    2 \alpha \omega
    \,
     \big(\frac{r^2-r_0^2}{r_0^2}+\VtwoOrF
     \big)  \sin^2 (\theta) \,d\tau \, d\phi
 \, .
\end{eqnarray}
%
Therefore variations with respect to the angular momentum parameter $a$ are never in $L^2$.

\section{The master equation for vector perturbations}
 \label{s5VII17.7}

The object of this appendix is to justify the convergences of the mode-decomposition series in the vector sector. We thus consider the vector projection $h^V$ of $h $, which we decompose into a complete
(cf., e.g.,~\cite{Boucetta1999})
 set of vector harmonics $\mVI$:
\bel{8VIII18.1}
 h^V_{ab}=0
 \,,
 \
 h^V_{ai} = r \sum_I \faVI \mVI
 \,,
 \
  h_{ij}^V= -  r^2 \sum_{I \,: \, k_V(I)\ne 0} \HTVI \frac 1 {k_V(I)}(\hat D_{i} \mVIj  + \hat D_{j} \mVI )
  \,,
\ee
where $k_V= k_V(I)$ in the last sum is determined by the corresponding eigenvalue of the vector Laplacian $\hat \Delta $ acting on $\mVI$:
\bel{8VIII18.2}
 \hat \Delta \mVI = - k_V^2(I) \mVI
 \,.
\ee
For $k\in \N$ let  $H^k(\Nman) $ denote the space of tensor fields on $\Nman$ of Sobolev regularity with $k$ derivatives. Standard functional analysis shows that we have
\bel{6IX17.21}
 \| D_{i_1}\cdots D_{i_\ell} \mVI \|_{L^2(\Nman)} \approx (1+k_V(I) )^\ell
 \,,
\ee
where we use $\approx$ to denote equivalence of norms,
hence
\bel{8VIII18.5}
  \| h^V_{ai}dx^i\|_{H^k(\Nman)}^2  \approx r^{2} \sum _I (1+ k_V(I) ^2)^ {k } |\faVI|^2
  \,.
\ee
Similarly, for any $j$,
\bel{8VIII18.5+}
    \|D_{a_1 \ldots a_j}  (r^{-1}h^V_{ai})dx^i\|_{H^k(\Nman)}^2  \approx  \sum _I (1+ k_V(I) ^2)^ {k } |D_{a_1 \ldots a_j}\faVI|^2
  \,.
\ee

The Ishibashi-Kodama master functions $\Phi_{V,I}$ are  defined, for
$$
 k_V(I)> (n-1)K
 \,,
$$
as solutions of the (integrable)
system (cf., e.g.,~\cite[Equations~(2.13)-(2.15)]{IshibashiKodamaStability})
  of PDE's
\begin{equation}\label{8VIII17.3}
  \partial_b(r^{n/2}\Phi_{V,\modeindex} )  =
  -r^{n-1}\epsilon_{b}{}^{a}
      \big(
       \faVI + \frac r {k_V(I)} D_a \HTVI
       \big)
       \,.
\end{equation}
Note that $r\ge r_0>0$ throughout, where $r_0$ is the location of the event horizon, so there is no issue of singularities arising in the equations  at $r=0$ in the current case.

If $\Phi_{V,I}$ is known we have, formally,
\beal{8VIII18.4}
 \lefteqn{
  \phantom{xxx}
   D_a (r^{-2} h_{ij}^V)
  =
    -   \sum_{I \,: \, k_V(I)\ne 0}
   \frac 1 {k_V(I)}D_a \HTVI(\hat D_{i} \mVIj  + \hat D_{j} \mVI )
   }
&&
\\
  &  =  &
         -   \sum_{I \,: \, 0<  \ k_V(I) \le (n-1)\kappahere }
   \frac 1 {k_V(I)}D_a \HTVI(\hat D_{i} \mVIj  + \hat D_{j} \mVI )
\nn
\\
  &    &
 +
   r^{-1}  \sum_{I \,: \, k_V(I)> (n-1)\kappahere }
  \big( \faVI
  -   {r^{-n+1}}
     \epsilon_a{}^b D_b ( \Phi_{V,\modeindex}
    \big)
      (\hat D_{i} \mVIj  + \hat D_{j} \mVI )
  \,.
   \nn
\eea

It is convenient to define the parts  $ \hVstar{ij}$, respectively  $\hVstar{ai}$, of $h_{ij}^V$, respectively of  $ \hVstar{ai}$, in which the low vector harmonics have been removed:
\begin{eqnarray}
 \label{6IX17.25}
 &
 \phantom{xx}
  \hVstar{ij}:= h_{ij}^V +r^2  \sum_{I \,: \, 0<  \ k_V(I) \le (n-1)\kappahere }
   \frac 1 {k_V(I)} \HTVI(\hat D_{i} \mVIj  + \hat D_{j} \mVI )
   \,,
   &
\\
 \label{6IX17.5}
 &
  \hVstar{ai}:= h_{ai}^V - r \sum_{I \,: \, 0<  \ k_V(I) \le (n-1)\kappahere }
    \faVI \mVI
   \,.
   &
\end{eqnarray}

From now on we assume that $\Phi_{V,I}$ vanishes.  From \eq{8VIII18.4} we then obtain
\bel{8VIII18.06}
 \| D_a (r^{-2} \hVstar{ij}) dx^i dx^j  \|_{H^k(\Nman)}
 \le 2 r^{-2}  \|  \hVstar{ai} dx^i  \|_{H^{k+1}(\Nman)}
  \,,
\ee
which furthermore justifies the convergence and equality \eqref{8VIII18.4} in the $\Phi_{V,I}\equiv 0$ case. Set, again formally
\begin{equation}\label{8VIII18.7}
  \YVstar{i}:= -r^2\sum_{I \,: \, k_V(I)> (n-1)\kappahere } \frac{\HTVI}{k_V(I)} \mvI
  \,.
\end{equation}
Then
\begin{equation}\label{8VIII18.07+}
  \| \YVstar{i} \|_{H^k(\Nman)}  =  \| \hVstar{ij}\|_{H^{k-1}(\Nman)}
  \,,
\end{equation}
which justifies convergence in \eqref{8VIII18.7}. Next, again formally
\begin{equation}\label{8VIII18.8}
  D_a\YVstar{i}  =
   -r^2\sum_{I \,: \, k_V(I)> (n-1)\kappahere } \frac{D_a\HTVI}{k_V(I)} \mvI =
   -r\sum_{I \,: \, k_V(I)> (n-1)\kappahere } \faVI \mvI
    =
    \hVstar{ai}
  \,.
\end{equation}
Hence
\begin{equation}\label{8VIII18.9}
 \phantom{xxx}
  \|D_a\YVstar{}  \|_{H^k(\Nman)}
 =
   \| \hVstar{ai}dx^i\|_{H^k(\Nman)}
  \,.
\end{equation}
This justifies convergence in \eqref{8VIII18.8}, and thus
\begin{equation}\label{8VIII18.8+}
    \hVstar{ai}=
    D_a\YVstar{i}  = \mcL_{\YVstar{}} \zgriem_{ai}
  \,.
\end{equation}
%


%

\bigskip

\noindent{\sc Acknowledgements:} The research of PTC was supported in
part by the Austrian Research Fund (FWF), Project  P29517-N27, and by the Polish National Center of Science (NCN) under grant 2016/21/B/ST1/00940.
 ED  was supported by the grant ANR-17-CE40-0034 of the French National Research Agency ANR (project CCEM).
PK was supported by a uni:docs grant of the University of Vienna.
Useful discussions with Jacek Jezierski and Olivier Sarbach are acknowledged.  We are grateful to Piotr Waluk and Jacek Jezierski for making their results available prior to publication.

\bibliographystyle{amsplain}

\def\polhk#1{\setbox0=\hbox{#1}{\ooalign{\hidewidth
  \lower1.5ex\hbox{`}\hidewidth\crcr\unhbox0}}} \def\cprime{$'$}
  \def\cprime{$'$} \def\polhk#1{\setbox0=\hbox{#1}{\ooalign{\hidewidth
  \lower1.5ex\hbox{`}\hidewidth\crcr\unhbox0}}} \def\cprime{$'$}
  \def\cprime{$'$}
\providecommand{\bysame}{\leavevmode\hbox to3em{\hrulefill}\thinspace}
\providecommand{\MR}{\relax\ifhmode\unskip\space\fi MR }
\providecommand{\MRhref}[2]{%
  \href{http://www.ams.org/mathscinet-getitem?mr=#1}{#2}
}
\providecommand{\href}[2]{#2}


\end{document}